\documentclass[11pt]{article}
\usepackage[english]{babel}
\usepackage{amsfonts}
\usepackage{amsthm}
\usepackage{amssymb}
\usepackage[dvips]{graphicx,color}
\usepackage{graphicx}
\usepackage{xcolor}
\usepackage{amsmath}
\usepackage{indentfirst}
\usepackage[dvips]{graphicx,color}
\usepackage{authblk}
\usepackage{enumitem}
\usepackage[margin=1.3in]{geometry}
\usepackage[title]{appendix}
\usepackage[normalem]{ulem}
\newcommand{\M}{{\mathcal{M}}}
\newcommand{\pM}{{\partial M}}

\newcommand{\tauC}{\tau_{chr}}

\newcommand{\Lchr}{L_{chr}}

\newcommand{\Lhaus}{L_{H}}

\newcommand{\Liminf}{{\rm lim\,inf}^o}
\newcommand{\Limsup}{{\rm lim\,sup}^o}

\newcommand{\qcd}{\begin{flushright} $\Box$ \end{flushright}}
\def\be{\begin{eqnarray}}
\def\ee{\end{eqnarray}}
\def\beq{\begin{equation}}
\def\eeq{\end{equation}}

\def\({\left (}
\def\){\right )}

\newtheorem{theorem}{Theorem}[section]
\newtheorem{lemma}[theorem]{Lemma}
\newtheorem{proposition}[theorem]{Proposition}
\newtheorem{corollary}[theorem]{Corollary}
\newtheorem{definition}{Definition}[section]
\newtheorem{remark}[theorem]{Remark}

\title{Hausdorff closed limits and the causal boundary of globally hyperbolic spacetimes with timelike boundary}
\author[1]{I.P. Costa e Silva}
\author[2]{J.L. Flores}
\author[3]{J. Herrera}
\affil[1]{\small{{\it Department of Mathematics\\
Universidade Federal de Santa Catarina, 88.040-900 Florian\'{o}polis-SC, Brazil.}}}
\affil[2]{\small{{\it Departamento de \'Algebra, Geometr\'{i}a y Topolog\'{i}a\\ Facultad de Ciencias, Universidad de M\'alaga\\ Campus Teatinos, 29071 M\'alaga, Spain.}}}
\affil[3]{\small{{\it Departmento de Matem\'aticas\\ Edificio Albert Einstein, Universidad de C\'ordoba\\ Campus de Rabanales, 14071 C\'ordoba, Spain.}}}
\begin{document}

\maketitle

	\begin{abstract}
\noindent
We show that when a spacetime $\M(=M\cup\partial M)$ is globally hyperbolic with (possibly empty) smooth timelike boundary $\partial M$, a {\em metrizable} topology, the {\em closed limit topology} (CLT) introduced by F. Hausdorff himself (see \cite{1963Hausdorffset}) in set theory, can be advantageously adopted on the Geroch-Kronheimer-Penrose causal completion of $\M$, retaining essentially all the good properties of previous topologies in this ambient.
In particular, we show that if the globally hyperbolic spacetime $\M$ admits a {\em conformal} boundary, defined in such broad terms as to include all the standard examples in the literature, then the latter is {\em homeomorphic} to the causal boundary endowed with the CLT. We also discuss how our recent proposal \cite{Costa_e_Silva_2018} for a definition of null infinity
using only causal boundaries can be translated when using the CLT, simplifying a number of technical aspects in the pertinent definitions.
In a more technical vein, in the appendix we discuss the relationship of the CLT with the more generally applicable (but not always Hausdorff) {\em chronological topology} \cite{Florescausalboundaryspacetimes2007, Floresfinaldefinitioncausal2011}, and show that they coincide exactly in those cases when the latter happens to be Hausdorff.

	\end{abstract}

\section{Introduction}

Since its introduction by Geroch, Kronheimer and Penrose in 1972 \cite{GerochIdealPointsSpaceTime1972}, the {\em causal boundary} (or {\em $c$-boundary}) of a given strongly causal spacetime $M$ has asserted itself as a method for attaching ideal points to a spacetime $M$ which, although somewhat more abstract than conformal boundaries, is (unlike the latter) deeply ingrained in the very causal structure of $M$. In particular, this means that the $c$-boundary is, in a strong sense, an inescapable feature of the (conformally invariant aspects of the) spacetime geometry.

By contrast, from a purely geometric perspective, the existence of a conformal boundary imposes an {\em ad hoc} restriction on $M$: it relies on the existence of a suitable open conformal embedding into a larger spacetime $\tilde{M}$, inducing a piecewise $C^1$ (often required to be $C^{\infty}$) boundary $\partial M$ with fairly restrictive properties (see, e.g., \cite{HawkingLargeScaleStructure1975, PenroseAsymptoticStructure1963, PenroseConformalInfinity1964, WaldGeneralRelativity1984, Frauendiener_2000} for general accounts). Often it either does not exist altogether, or else has counterintuitive properties, even in simple but key examples such as plane waves \cite{berenstein02:super_yang_mills, MarolfPlanewavesinfinity2002, Hubeny_2002, Florescausalboundarywavetype2008}. Uniqueness is also unknown in general (although it has been established in some particular instances; see, e.g., \cite{ChruscielConformalboundaryextensions2010}).

It cannot be gainsaid that conformal boundaries are extremely convenient in physical applications. Whenever available, they present a very elegant and well-motivated means of addressing a number of key physical notions, such as asymptotic flatness, black holes and gravitational waves \cite{Frauendiener_2000,HawkingLargeScaleStructure1975, WaldGeneralRelativity1984}, and indeed it is hard to see how to proceed in a physically sensible fashion otherwise. The original causal boundary construction \cite{GerochIdealPointsSpaceTime1972}, by contrast, presented a number of undesirable properties even in the simplest cases.

Nevertheless, it can be of great interest to take advantage of the universality and geometric naturalness of the causal boundary construction, while making it as amenable as possible to concrete applications. Fortunately, the glitches in the original proposal have been addressed and solved by a suitable redefinition of the $c$-boundary, initially by Marolf and Ross \cite{Marolfnewrecipecausal2003}, and then by Flores, Herrera and S\'{a}nchez \cite{Florescausalboundaryspacetimes2007, Floresfinaldefinitioncausal2011} (see also the references of the latter papers and in \cite{Garcia-ParradoCausalstructurescausal2005} for a detailed discussion of previous/alternative proposals by other authors).

It may be too restrictive to assume that, upon adding a $c$-boundary, the extended space be a smooth manifold, at least in general cases. But having a good topology on the extension is fundamental in the applications of {\em any} boundary proposal. However, it often turns out that we cannot even hope that the extended space is {\em Hausdorff} without violating other fundamental good features of the $c$-completion ($M$ plus its $c$-boundary). Indeed, as emphasized by S. Harris, this non-Hausdorff character seems to be an inescapable feature of $c$-boundaries. Accordingly, the otherwise quite functional {\em chronological topology} \cite{Florescausalboundaryspacetimes2007, Floresfinaldefinitioncausal2011} fails to be Hausdorff in many cases - even in globally hyperbolic examples - although they are often $T_1$. However, we feel that the question still deserves to be investigated whether the Hausdorff property can be retained for a (natural) topology on the $c$-completion - perhaps with more restrictive condition on the spacetimes - without detracting from other useful features.


The main purpose of the present work is to show that in the class of globally hyperbolic spacetimes with a smooth timelike boundary it is indeed possible to endow the so-called {\em future} $c$-completion with a very natural, {\em metrizable} (and in particular Hausdorff) topology from the outset. This topology turns out to be strictly finer than the chronological topology in \cite{Florescausalboundaryspacetimes2007, Floresfinaldefinitioncausal2011}, but retains all its good properties.

For that purpose, we adapt a well-known topology introduced on the set $\mathcal{C}_X$ of closed, non-empty subsets of a metric space $X$ by F. Hausdorff (see \cite{1963Hausdorffset}), and which we refer to as the {\em closed limit topology} (CLT for short). This topology boasts an impressive list of good services in geometry. For instance, it is used in the Hausdorff-Gromov convergence of Riemannian manifolds \cite{Petersen_1993} which, among other things, is the starting point of the theory of Alexandrov spaces used by G. Perelman in his proof of the Poincar\'{e} conjecture \cite{Perelman_2008}.

As a rule, however, sequences of even very nice sets, such as compact smooth submanifolds, fail to have regular limits in the CLT in the sense that certain key geometric properties may not pass to the limit. But with some extra ``controlling features'', the limit may be regular. In a Lorentzian manifold $(M,g)$, it turns out that the {\em causal properties} of certain subsets thereof gives such a controlling feature. For example, the CLT is closely related to the $C^0$ topology on the space of inextendible causal curves, which is key in many applications. Another example much closer to our context is the recent observation by G. Galloway and C. Vega \cite{GallowayAchronalLimitsLorentzian2014, GallowayHausdorffClosedLimits2017} that {\em achronal boundaries are preserved by Hausdorff closed limits}. The CLT also works well (and is so natural) on $c$-boundaries of globally hyperbolic spacetimes with timelike boundary precisely because the so-called {\em indecomposable past sets} (IPs) we use to define these boundaries can also be controlled in this setting.

The rest of this paper is organized as follows.

In section \ref{s2} we recall the notion of spacetime with smooth timelike boundary, and review the main features on these objects used later in the paper.

In section \ref{prelim}, we briefly review some facts about {\em chronological sets}, not only to establish terminology, but also to describe the construction of a (future) $c$-boundary as a set endowed with a chronology relation. We then state {\em Harris' universality theorem}, which provides a precise way to characterize the uniqueness of putative boundaries with minimal causal properties, even outside the $c$-boundary paradigm.

In sections \ref{top1} and \ref{top2} we show how the CLT can be introduced on the $c$-boundaries, and prove that it retains most good properties of the chronological topology.

In section \ref{conf}, we compare our construction with conformal boundaries (suitably defined). In particular, we show that the future $c$-boundary endowed with the CLT is homeomorphic to the part of the conformal boundary which is accessible by future-directed timelike curves. A similar result has also been obtained for the chronological topology in \cite{Floresfinaldefinitioncausal2011} and \cite{OlafConformal}. The approach we adopt in this section has been inspired by the one in \cite{OlafConformal}: we essentially borrow from that reference the notions of a {\em future-precompact} subset and {\em future-nesting} conformal extensions we use in section \ref{conf}, although a somewhat different terminology is used here to describe these\footnote{Indeed, in Ref. \cite{OlafConformal} a Hausdorff topology was also introduced on the $c$-boundary following a suggestion in \cite{GerochIdealPointsSpaceTime1972}. However, that topology is introduced in a fairly technical manner, and its relation to the CLT, as well as its naturalness, remain unclear to us.}.

In section \ref{Scriplus} we revisit in the present context the notion of {\em future null infinity} for $c$-boundaries, recently introduced by us in a companion paper \cite{Costa_e_Silva_2018} with an eye towards developing some aspects of the theory of black holes within the $c$-boundary paradigm. In \cite{Costa_e_Silva_2018} we aimed at maximum generality, but at the price of introducing some rather technical-looking definitions and results. Reinterpreted in the light of the CLT topology on the $c$-boundary,
the underlying discussion becomes (hopefully) more transparent. In particular, it is possible to carry out a closer comparison with conformal boundaries and see how the concepts developed in \cite{Costa_e_Silva_2018} apply in this case.

In order to address some specific, more technical aspects of the CLT which are nevertheless somewhat detached from the main results listed above, we have included one appendix in the paper. There, we discuss in more detail the relationship of the CLT with the chronological topology. In particular, we prove that the CLT coincides with the chronological topology if and only if the latter is Hausdorff.


Finally, we emphasize that all the results in this paper are purely geometric, i.e., independent of any field equations for the metric, and work for all dimensions equal to or greater than 2.

\section{Spacetimes with timelike boundary}\label{s2}

Throughout this paper, we consider spacetimes with a (possibly empty) {\em smooth timelike boundary}. Since some finer aspects of this topic may be unfamiliar to the reader, we review the needed material in this section. We shall assume, however, that the reader is familiar with standard facts in the causal theory of spacetimes without boundary as given in the basic references \cite{BeemGlobalLorentzianGeometry1996, HawkingLargeScaleStructure1975, ONeillSemiRiemannianGeometryApplications1983}.

\subsection{Generalities on spacetimes with boundary} \label{s2.1}

{\bf Manifolds with boundary.} In what follows $\M$ will denote a {\em connected} $C^{\infty}$
{\em  $n$-manifold with boundary} ($n\geq 2$).
Any functions or tensor fields will be tacitly assumed to be smooth (by which we mean $C^{\infty}$), unless otherwise stated.
 $\M$ is then  locally diffeomorphic to (open subsets of) a closed half space of ${\mathbb R}^n$;
 $M$ will denote its {\em interior} and $\pM$ its 
 {\em boundary}.
  For any $p\in \M$,   $T_p\M$ will denote
 its $n$-dimensional tangent space
 while for $p\in\pM$,  $T_p\pM$ is the
 $(n-1)$-dimensional tangent space to the boundary.
 Background on  manifolds with boundary can be found, e.g., in \cite[Section 9]{Lee_2003}.




\bigskip

\noindent {\bf Spacetimes with boundary.} 
Most of the usual notions for Lorentzian manifolds without boundary, such as metric tensors, connection, curvature or causal character of vectors (following the conventions in \cite{BeemGlobalLorentzianGeometry1996, Minguzzicausalhierarchyspacetimes2008}), are extended to the case with boundary with no notable differences (see \cite{Gallowaysummer_2014, SolisPHD} for further background.) Therefore, we shall mostly focus below on the eventual differences that arise in this case.

\medskip

We focus only on the case of timelike boundaries here, and merely state the definitions and results which are pertinent to this paper. We refer the reader especially to Refs. \cite{SolisPHD,AkeFloresSanchezTimelikeBoundary2018} for more detailed discussions and proofs.

\begin{definition}
A {\em Lorentzian manifold with timelike boundary} $(\M,g)$, $\M=M\cup \partial M$, is a Lorentzian manifold with boundary such that, for the natural inclusion $i:\partial M \hookrightarrow \M$, the pullback $i^{*}g$ is a Lorentzian
metric on the (possibly empty or non-connected) boundary. A {\em spacetime with timelike boundary} is a connected, time-oriented Lorentzian manifold with timelike boundary.
\end{definition}
\noindent By {\em time-oriented} we of course mean that a time cone has been chosen on the tangent space at each point by some globally defined, continuous timelike vector field $Z$ on $\M$, just as in the case of spacetimes without boundary. In particular, we shall assume that the interior $M$ and the boundary $\partial M$ are also time-oriented by the corresponding restrictions of $Z$, so that $M$ and each connected component of $\partial M$ become spacetimes (without boundary) on their own right with the restricted metric. We shall always denote by $(M,g)$ the interior spacetime, retaining the same notation for the metric $g$ when restricted to $M$ if there is no risk of confusion. (Of course, $(\M, g) \equiv (M,g)$ if $\partial M = \emptyset$.) The pull-back $i^*$ will be dropped when there is no possibility of confusion; the time-orientation of vectors, vector fields and causal curves are defined just as in spacetimes without boundary.

\bigskip

\noindent {\bf Causality in spacetimes with timelike boundary.} For a spacetime with timelike boundary $(\M,g)$, the usual notations  $\ll$, $\leq$, $I^\pm(p)$, $J^\pm(p)$, etc., will be used for the chronological and causal relations (often with a subscript `$g$' for clarity), and for the chronological and causal future/past of any $p\in \M$, with no change in their meaning. In particular, $I^+(p,U)$, for $p\in U$, will denote the chronological future obtained by using (piecewise smooth) timelike curves entirely contained in the subset $U\subset \M$.

However, an important technical difference arises here for spacetimes with non-empty boundary. Already in the context of spacetimes without boundary it is well-known that limits of sequences of piecewise smooth causal curves (with suitably defined topologies on sets of causal curves) are usually not piecewise smooth. Accordingly, one must extend the notion of causal curves to include curves which are merely continuous (cf. \cite[Section 3.2]{BeemGlobalLorentzianGeometry1996}.) The naturally extended definitions of causal future and past, however, coincide with those using piecewise smooth curves. An analogous extension does exist for spacetimes with timelike boundary \cite[Def. 3.19]{SolisPHD}, and again causal past and futures are left unchanged thereby, but now it becomes unclear if the Limit Curve Lemma, which is crucial for a number of constructions in causal theory, holds for this class of curves. (The proof presented in Prop. 3.24 of Ref. \cite{SolisPHD} seems to be incomplete - cf. Remark 2.17 in Ref. \cite{AkeFloresSanchezTimelikeBoundary2018}.) In order to circumvent this problem, it seems necessary to work with curves in a Sobolev class (causal $H^1$-curves - see \cite[Section 2.3]{AkeFloresSanchezTimelikeBoundary2018} and references therein.) However, {\em for globally hyperbolic spacetimes with timelike boundary} (defined below - this is the only case of interest for us in this paper), causal past and future defined using this class of curves do coincide with the usual ones defined using piecewise smooth causal curves (cf. \cite[Prop. 2.19]{AkeFloresSanchezTimelikeBoundary2018}). Hence, this issue will never affect us here and will be consistently ignored throughout.


In what follows, the symbol `$cl$' before a subset of a given topological space will always denote its (topological) closure therein.

The following proposition establishes the basic properties for the chronological and causal relations in spacetimes with a timelike boundary.
\begin{proposition}\cite[Prop. 2.6]{AkeFloresSanchezTimelikeBoundary2018}\label{p_opentrans} In a spacetime with timelike boundary, the following facts hold.
\begin{itemize}
\item[a)] The binary
relation $\ll$ is open (in particular, all $I^\pm(p)$ are open in $\M$).
\item[b)] For any $p,q,r\in \M$, $p\ll q\leq  r \Rightarrow p\ll r$,
$p\leq q\ll  r \Rightarrow p\ll r$.
\item[c)]$J^\pm(p)\subset cl(I^\pm(p))$ for each $p \in M$.
\item[d)]  $I^{\pm}(p,M)=I^{\pm}(p) \cap M$ for each $p \in M$.
\end{itemize}
\end{proposition}

\begin{remark}\label{nuevo} {\em Property (d) can be extended to points at the boundary as follows: for any $p\in \partial M$, $I^{\pm}(p)\cap M$ is the set of
$q\in M$ such that there exists a  future/past -directed timelike $\gamma$ with ${\rm Im}(\gamma)\setminus\{p\}\subset M$ joining $p$ with $q$. }
\end{remark}

\subsection{The causal ladder for spacetimes with timelike boundary}\label{s2.2}

Most of the causal ladder and characterizations in the case without boundary (see \cite{BeemGlobalLorentzianGeometry1996, Minguzzicausalhierarchyspacetimes2008}) can be transplanted directly to the case with timelike boundary. 
Here, we will recall some of them.

\begin{definition}\cite[Defn. 2.8]{AkeFloresSanchezTimelikeBoundary2018}
 A spacetime with timelike boundary $(\M,g)$
is:
\begin{itemize}\item
{\em chronological} (resp.
 {\em causal})
 if it does not contain closed
 timelike (resp.
closed
causal)
curves,
\item {\em future} (resp.  {\em past})  {\em distinguishing}  if the equality  $I^{+}(p)=I^{+}(q)$ (resp. $I^{-}(p)=I^{-}(q)$) implies $p=q$ for all $p,q \in \M$; $(\M,g)$ is {\em distinguishing} when it is both future and past distinguishing.
\item {\em strongly causal} if for all $p\in \M$ and  any neighborhood $U\ni p$ there exists another neighborhood $V \subset U$, $p\in V$, such that any causal curve with endpoints in
$V$ is entirely contained in $U$.
\end{itemize}
\end{definition}


There is another important characterization of strong causality, namely in terms of {\em causally convex} sets, that also holds for spacetimes with timelike boundary.

\begin{definition}\cite[Defn. 2.9]{AkeFloresSanchezTimelikeBoundary2018} A subset $W$ of a spacetime with timelike boundary $(\M,g)$ is
{\em  causally convex 
 } if $J^{+}(x)\cap J^{-}(y)
\subset W$ for any $x,y \in W$ (equivalently, if any causal curve with endpoints in $W$ remains in $W$).


Given an open, connected neighborhood $U\subset \M$, a subset $W\subset U$ is {\em causally convex}
{\em in $U$} when $W$ is causally convex  
 as a subset of $U$ when we regard $(U, g|_U)$ as a spacetime with timelike boundary.
\end{definition}

Thus we have:
%

\begin{proposition}\cite[Prop. 2.11]{AkeFloresSanchezTimelikeBoundary2018} \label{imprisoned}
$(\M,g)$ is  strongly causal  if and only if for each $p\in\M$ and for any neighbourhood
$U\ni p$, there exists a causally convex
neighbourhood $W\ni p$ contained in $U$.

In this case,
causal curves are not partially imprisoned on compact sets, that is, for any  future- or past-inextendible causal curve $\gamma:[a,b) \rightarrow \M, a<b\leq \infty$ and any compact set $K\subset \M$, there exists some $s_0\in [a,b)$ such that $\gamma(s)\not\in K$ for all $s\in [s_0,b)$.

\end{proposition}

We now proceed to the higher steps of the causal ladder:
\begin{definition}\cite[Defn. 2.14]{AkeFloresSanchezTimelikeBoundary2018}
 \label{higherladder}
A spacetime with timelike boundary $(\M,g)$
is:
\begin{itemize}
\item {\em stably causal} if it admits a  {\em time function}, i.e., a continuous function $\tau$ which increases strictly on all future-directed causal curves.
\item {\em causally continuous} if the set valued functions $I^{\pm}:\M \rightarrow \mathbb{P}(\M)$ are both one-to-one and continuous (for the natural topology on the power set $\mathbb{P}(\M)$ which admits as a basis the sets  $\{U_K: K\subset \M$ is compact$\}$, where $U_K=\{A\subset \M: A\cap K=\emptyset\}$, see \cite[Def. 3.37 to Prop. 3.38]{Minguzzicausalhierarchyspacetimes2008});
\item {\em causally simple} if it is distinguishing and $J^{\pm}(p)$ are closed for every $p \in \M$;
\item {\em globally hyperbolic}, when it is strongly causal and all $J^+(p)\cap J^-(q)$, $p,q\in \M$ are compact.
\end{itemize}
\end{definition}

The following result summarizes the steps of the causal ladder for spacetimes with timelike boundary.
\begin{proposition}\cite[Props. 2.12, 2.13, Theor. 3.8]{AkeFloresSanchezTimelikeBoundary2018}\label{p_ord_lowerlevels} In any spacetime $(\M, g)$ with timelike boundary: global hyperbolicity $\Rightarrow$ causal simplicity $\Rightarrow$ causal continuity $\Rightarrow$ stable causality $\Rightarrow$ strong causality $\Rightarrow$ distinguishing, and future or past distinguishing $\Rightarrow$ causality $\Rightarrow$ chronology. Moreover, if $(\M,g)$ is globally hyperbolic, then its interior $(M,g)$ is causally continuous (but not necessarily causally simple).
\end{proposition}

\section{Chronological sets and the universality of the future $c$-boundary}\label{prelim}

S. Harris \cite{HarrisUniversalityfuturechronological1998} has shown that the future (or past) $c$-completion of a large class of spacetimes always exists and is unique in the following very strong sense. A category can be defined whose objects are {\em chronological sets} - a class which naturally encompasses chronological spacetimes - and whose morphisms are certain {\em chronology-preserving maps} (these definitions are briefly reviewed below). A natural notion of {\em (future or past) completeness} is given, and for a certain (full) subcategory whose objects include {\em globally hyperbolic spacetimes with timelike boundary}, the existence of a future (or alternatively a past) $c$-completion is given by a universal property. Since these notions are so general as to be satisfied for nearly any conceivable extension of spacetimes with sensible causal properties, this establishes the essential uniqueness of the $c$-completion.

In the first part of this section we review some basic facts about chronological sets which will be relevant to our discussion, for the benefit of the reader unfamiliar with the pertinent notions. We refer to Ref. \cite{HarrisUniversalityfuturechronological1998} and references therein for further details and proofs.

\begin{definition}
\label{chronos}
A {\em chronological set} is a pair $(X,\ll)$, where $X$ is a non-empty set and $\ll$ is a strict order on $X$ (i.e., a transitive and antisymmetric binary relation) such that
\begin{itemize}
\item[C1)] $\forall x \in X$, $\exists y \in X$ such that either $x \ll y$ or $y \ll x$, and
\item[C2)] ({\em Chronological separability}) there exists a countable set $S \subset X$ such that $ \forall x,y \in X$ with $x\ll y$, there exists $s \in S$ such that $x \ll s \ll y$. (Such a set $S$ is said to be {\em dense} in $X$.)
\end{itemize}
\end{definition}


Throughout the remaining of this section, we shall fix a chronological set $(X,\ll)$.  The following definitions naturally extend those introduced for spacetimes in the original reference \cite{GerochIdealPointsSpaceTime1972} and will be used in all that follows.

\begin{definition}
\label{shebang}
Let $A \subseteq X$ be any set.
\begin{enumerate}[label=(\roman*)]
\item The {\em past} [resp. {\em future}] of $A$ is
\[
I^-(A) \, [\mbox{resp. } I^+(A)] := \{ x \in X \, : \, \exists y \in A \mbox{ such that } x \ll y \, [\mbox{resp. } y \ll x] \}.
\]
(We shall write $I^{\pm}(x) := I^{\pm}(\{x\})$ for any $x \in X$.)
\item $A$ is a {\em past set} [resp. {\em future set}] if $A = I^-(A)$ [resp. $A = I^+(A)$].
\item $A$ is an {\em indecomposable past set (IP)} [resp. {\em indecomposable future set (IF)}] is it is a past [resp. future] set and cannot be written as the union of two proper past [resp. future] subsets.
\item A point $x \in X$ is a {\em future bound} [resp. {\em past bound}] for $A$ if $A \subset I^-(x)$ [resp. $A \subset I^+(x)$]. If a future [resp. past] bound for $A$ exists, then $A$ is said to be {\em future-bounded} [resp. {\em past-bounded}].
\item A point $x \in X$ is a {\em future limit} [resp. {\em past limit}] for $A$ if it is a future [resp. past] bound with the following property:
\[
y \ll x  \mbox{ [resp. $x \ll y$]} \Longrightarrow \exists a \in A \mbox{ such that } y \ll a \mbox{ [resp. $a \ll y$]}.
\]
\item $A$ is said to be a {\em proper indecomposable past set (PIP)} [resp. {\em proper indecomposable future set (PIF)}] if $A$ is an IP [resp. IF] and admits a future limit [resp. past limit]. $A$ is said to be a {\em terminal indecomposable past set (TIP)} [resp. {\em terminal indecomposable future set (TIF)}] if it is an IP [resp. IF] which is not a PIP [resp. TIP].
\item $A$ is {\em achronal} if $x \notin I^+(y), \forall x,y \in A$.
\end{enumerate}
\end{definition}

\begin{remark}
\label{matteroffact}
{\em The following facts are immediate consequences of Definition \ref{shebang} for a chronological set $(X,\ll)$.
\begin{enumerate}[label=(\arabic*)]
\item The collection of IPs in $(X,\ll)$ is the disjoint union of the collections of PIPs and TIPs. A ``time-dual'' conclusion can be made for IFs.
\item The empty set $\emptyset \subset X$ is vacuously both a TIF and a TIP. It may well be the case, for a general chronological set, that there are no other TIPs or TIFs. For example, we may take $X$ to be the closed strip $-1 \leq t \leq 1$ in the Minkowski plane $(M=\mathbb{R}^2, g= -dt^2 \oplus dx^2)$, with $\ll$ being the restricted chronological order arising from $\ll_g$. In this case there no non-empty TIPs or TIFs.
\item A future or past bound of a subset $A\subset X$ can never be in $A$.
\item $X$ itself is a past [resp. future] set if and only if $I^+(x) \neq \emptyset$ [resp. $I^-(x) \neq \emptyset$] for all $x \in X$. If $X$ is a past [resp. future] set, then it may or may not be indecomposable. If it is an IP [resp. IF], however, then it is necessarily a TIP [resp. TIF] in view of the previous observation.
\end{enumerate}
}
\end{remark}

In what follows we shall deal mostly with IPs, past sets, future bounds, etc.. {\em The analogue results for IFs, future sets, past bounds, etc. will be always understood unless it is explicitly stated otherwise.}

We now turn to the notion of {\em future completion} of a chronological set. It turns out that not all chronological sets have a suitable future completion, so we shall define a pertinent category of such sets for which future completions have good categorical properties.

To simplify the notation, and following standard usage, we shall often abuse notation and denote a chronological set $(X,\ll)$ by its underlying set $X$, referring to $X$ itself as the chronological set (in which the relation $\ll$ is implicit).

\begin{definition}
\label{objects}
A chronological set $X$ is
\begin{enumerate}[label=(\arabic*)]
\item {\em past-determined} if  $\forall x,y,w \in X$,
\[
w \ll y,\;\, I^-(x) \neq \emptyset \mbox{ and } I^-(x) \subset I^-(w) \Longrightarrow x \ll y;
\]
\item {\em past-regular} if $\forall x \in X$, $I^-(x)$ is a non-empty IP;
\item {\em future-complete} if any {\em future-directed chain} in $X$ (i.e., any sequence $(x_n)\subset X$ with $x_n\ll x_{n+1}$ for all $n$) has a future limit, or equivalently, if every non-empty IP is a PIP.
\item {\em past} [resp. {\em future}] {\em distinguishing} if
\[
x \neq y \Longrightarrow I^-(x) \neq I^-(y) \mbox{ [resp. $I^+(x) \neq I^+(y)$]}.
\]
If $X$ is both future and past distinguishing, then it is said to be {\em distinguishing}.
\end{enumerate}
\end{definition}

Next, we introduce the morphisms which will be relevant in this context:
\begin{definition}
\label{morphisms}
A mapping $f: X \rightarrow Y$ between chronological sets is
\begin{enumerate}[label=(\arabic*)]
\item {\em chronological} (or {\em chronology-preserving}) if  $\forall x,x' \in X$,
\[
x \ll  x'  \Longrightarrow f(x) \ll  f(x');
\]
\item {\em future-continuous} if it is chronological and whenever $x \in X$ is the future limit of a future-directed chain $(x_k)$ in $X$, then $f(x)$ is the future limit of the future-directed chain $(f(x_k))$ in $Y$.
\end{enumerate}
\end{definition}


We shall denote by $\mathfrak{X}$ the category whose objects $Ob(\mathfrak{X})$ are {\em past-regular, past-determined, past-distinguishing chronological sets} and whose morphisms $Mor(\mathfrak{X})$ are with future-continuous mappings. This the relevant category, as we will see presently. Indeed, Proposition \ref{causalityforpastdetermination} below implies that causally continuous spacetimes without boundary and globally hyperbolic spacetimes with timelike boundary (with their natural chronological relations) are objects in $\mathfrak{X}$.

\begin{definition}
\label{completion}
A {\em future completion} of an object $X$ in $\mathfrak{X}$ is a pair $(\hat{X},i_X)$ where $\hat{X}$ is a future-complete object of $\mathfrak{X}$ and $i_X: X \rightarrow \hat{X}$ is a future-continuous one-to-one mapping satisfying the following universal property: given any future-complete $Z \in Ob(\mathfrak{X})$ and any future-continuous mapping $\varphi: X\rightarrow Z$ there exists a unique future-continuous mapping $\psi:\hat{X} \rightarrow Z$ for which $\psi \circ i_X =\varphi$.
\end{definition}

Since entails a universal property, such a definition ensures that {\em if a future completion exists, then it is unique apart from an isomorphism in $\mathfrak{X}$}. Thus we may refer to {\em the} future completion of a given chronological set.

We are now ready to present Harris' key result.
\begin{theorem}[Harris' Universality Theorem \cite{HarrisUniversalityfuturechronological1998}]
\label{crucial}
A future completion always exists for every object in the category $\mathfrak{X}$.
\end{theorem}
\noindent {\em Comment on the proof.} Let $X$ be an object in $\mathfrak{X}$. A concrete realization of the future completion is given as follows. Define
\[
\hat{X}:= \{ P \, : \, P \mbox{ is a non-empty IP in $X$}\},
\]
and $i_X:x \in X \mapsto I^-(x) \in \hat{X}$. Now, since $X$ is past-regular $i_X$ is well-defined, and since $X$ is past-distinguishing, it is also one-to-one. We define a chronological relation on $\hat{X}$ as follows: $\forall P,P' \in \hat{X}$
\[
P \ll P' \Longleftrightarrow \exists\; x \in P' \setminus P \mbox{ such that } P \subset I^-(x).
\]
It is easy to see this indeed defines a chronological relation. One then shows that chronological separability holds, and that $\hat{X}$ is past-regular, past-determined and past-distinguishing, and hence an object in $\mathfrak{X}$, and that $i_X$ is future-continuous with respect to this chronological relation. \qcd

We shall also need the following simple result.
\begin{proposition}
\label{ghlater}
 Let $i_X: X \hookrightarrow \hat{X}$ be the future completion of some chronological set $X\in Ob(\mathfrak{X})$. Let $Y$ be a chronological set and $f: \hat{X} \rightarrow Y$ a chronological bijection with chronological inverse. Then $Y$ is future-complete, $f,f^{-1} \in Mor(\mathfrak{X})$ and $(Y, f\circ i_X)$ is a future completion of $X$.
\end{proposition}
\begin{proof} Clearly, the fact that $f$ is chronological with chronological inverse means that $f$ preserves past and futures in the sense that
\[
f(I_{\hat{X}}^{\pm}(x)) = I_Y^{\pm}(f(x)),\quad \forall x \in \hat{X},
\]
which implies that $Y$ is also past-regular, past-determined, past-distinguishing, i.e., $Y \in Ob(\mathfrak{X})$. To show that it is future-complete, let $(y_n) \subset Y$ be any future-directed chain therein. Let $\hat{x}$ be the future limit of the future-directed chain $(f^{-1}(y_n)) \subset \hat{X}$. Then
\[
f^{-1}(y_n) \ll \hat{x} \Longrightarrow y_n \ll f(\hat{x}), \forall n \in \mathbb{N},
\]
and given $y \ll f(\hat{x})$, we have $f^{-1}(y) \ll \hat{x}$, and hence for large enough $n$ we have $f^{-1}(y) \ll f^{-1}(y_n)$, and thus
\[
y \ll y_n.
\]
We conclude that $f(\hat{x})$ is the future-limit of $(y_n)$. This proves that $Y$ is indeed future-complete as claimed, and also establishes (reasoning for $f^{-1}$ in an entirely analogous fashion), that $f,f^{-1} \in Mor(\mathfrak{X})$.

Finally, let $Z \in Ob(\mathfrak{X})$ be future-complete and let $g \in Hom_{\mathfrak{X}}(X,Z)$. Consider the unique morphism $k: \hat{X} \rightarrow Z$ such that $k \circ i_X = g$, and define $\tilde{k}: Y \rightarrow Z$ by
\[
\tilde{k}:= k \circ f^{-1}.
\]
Then clearly $\tilde{k} \in  Mor(\mathfrak{X})$ and $\tilde{k} \circ (f \circ i_X) \equiv g$. Moreover, if $j: Y \rightarrow Z$ is any other morphism such that $j \circ (f \circ i_X) = g$, then $(j \circ f) \circ i_X = g$, and hence $j \circ f \equiv k$, which in turn  implies that $j \equiv \tilde{k}$. This completes the proof.
\end{proof}

\medskip

We proceed to apply these notions and results to our context. First, note that if $(\M,g)$ is a {\em chronological spacetime with timelike boundary} (recall that the possibility $\partial \M=\emptyset$ is included here), i.e., there are no closed timelike curves therein, then $(\M,\ll _g)$ is a chronological set, where $\forall p,q \in M$,
\[
p \ll_g q\Longleftrightarrow \exists \mbox{ a future-directed timelike curve $\alpha: [0,1] \rightarrow \M$ with $\alpha(0) = p, \alpha(1) =q$.}
\]

In the original reference \cite{GerochIdealPointsSpaceTime1972} some basic facts for IPs were established in the context of spacetimes. We state these results here in the more general case of spacetimes with timelike boundary. (Adapting the proofs in \cite{GerochIdealPointsSpaceTime1972} for this case is straightforward and left as an easy exercise to the reader.)

\begin{proposition}
\label{PipST}
For a spacetime with timelike boundary
$(\M,g)$, the following holds.
\begin{enumerate}[label=(\arabic*)]
\item In a chronological spacetime with timelike boundary $(\M,g)$, a past set $P \subset \M$ is a PIP if and only if there exists $p \in \M$ such that $P =I^-(p)$. If $(\M,g)$ (with timelike boundary) is distinguishing, then $p$ is uniquely determined.
\item In a strongly causal spacetime with timelike boundary $(\M,g)$, a non-empty past subset $P \subset \M$ is a TIP if and only if there exists a future-inextendible timelike curve $\gamma:[0,a) \rightarrow \M$ such that
\[
P = I^-(\gamma).
\]
In the affirmative case we say that $\gamma$ {\em generates} $P$.
\item A non-empty set $P \subset \M$ is an IP if and only if there exists a {\em future-directed chain} in $P$, i.e., a sequence $(x_k) \subset P$ such that
\[
x_k \ll x_{k+1},\quad \forall k \in \mathbb{N}
\]
for which
\[
P = I^-(\{x_k \, : \, k \in \mathbb{N}\}).
\]
(In this case, we shall say that the future-directed chain $(x_k)$ {\em generates} $P$.) Equivalently, $P$ is an IP if and only if it is {\em directed}, i.e., if for any $x,y \in P$ there exists $z \in P$ such that $x,y \ll_g z$.

\item Let $P \subset X$ be a non-empty IP and $(x_k) \subset P$ a future-directed chain generating $P$. Then $x \in X$ is a future limit of $P$ if and only if it is a future limit of the set
\[
A= \{x_k \, : \, k \in \mathbb{N}\}.
\]
In the affirmative case,  $I^-(x) = I^-(A)$. Thus, $P$ is a TIP if an only if such a limit does not exist, in which case the future-directed chain generating it is said to be {\em future-endless}.
\item If $(\M,g)$ is strongly causal, $p \in \M$ is the future limit of the future-directed chain $c = (x_k)$ if and only if $x_k \rightarrow p$.

\end{enumerate}
\end{proposition}

Clearly, a strongly causal spacetime with timelike boundary $(\M,g)$ is always past-regular and past-distinguishing,
but it may or may not be past-determined. However, with suitable causality conditions, $(\M,g)$ will be past-determined. As already indicated, we shall be focussing largely on globally hyperbolic spacetimes in this paper, maybe with a non-empty timelike boundary. For this class of spacetimes, we have:
\begin{proposition}
\label{causalityforpastdetermination}
Suppose that a spacetime with timelike boundary $(\M,g)$ is either
\begin{enumerate}[label=(\alph*)]
\item \label{causpast-1} causally continuous with $\partial M=\emptyset$, or
\item \label{causpast-2}globally hyperbolic.
\end{enumerate}
Then, $(\M,\ll_g)$ is past-determined.
\end{proposition}
\begin{proof} Let $x,y,w \in M$ such that
\[
w \ll_g y,\;\, I^-(x) \neq \emptyset \mbox{ and } I^-(x) \subset I^-(w).
\]
Suppose by contradiction that \textit{\ref{causpast-1}} holds, but $x \notin I^-(y)$. Clearly $x \in {\rm cl}(I^-(y))$, so $x \in \partial I^-(y)$. Causal continuity then implies \cite[Lemma 3.42]{Minguzzicausalhierarchyspacetimes2008} that $y \in  \partial I^+(x)$. Since $w \in I^-(y)$ this means that $w \notin {\rm cl}(I^+(x))$. Again by causal continuity there exists $U \ni x$ open such that
\[
x' \in U \Longrightarrow w \notin {\rm cl}(I^+(x')).
\]
But since $U$ is open we can choose $x'  \in I^-(x)\cap U$, and then $x' \in I^-(w)$, i.e., $w \in I^+(x')$, a contradiction.

Now assume that \textit{\ref{causpast-2}} holds. Note that $x \in {\rm cl}(I^-(x)) \subset {\rm cl}(I^-(w)) = J^-(w)$ (Proposition \ref{p_ord_lowerlevels}) in this case. Then
\[
x \leq_g w \ll_g y \Rightarrow x \ll _g y
\]
by Proposition \ref{p_opentrans}. This concludes the proof.
\end{proof}

For spacetimes without boundary, smooth (time-order preserving) conformal embeddings are obvious examples of future-continuous mappings.

Since we have guaranteed uniqueness in our cases of interest, {\em we may as well work, and shall do so throughout the rest of this paper, with the concrete realization of the {\em future c-completion} $(\hat{\M}, \ll, i_{\M})$ of $(\M,\ll_g)$ given in the proof of Theorem \ref{crucial}. Explicitly,
\begin{equation}
\label{concrete1}
\hat{\M}:= \{ P \, : \, P \mbox{ is a non-empty IP in $\M$}\};
\end{equation}
for any $P,P' \in \hat{\M}$, the {\em chronological relation} on $\hat{\M}$ is defined as
\begin{equation}
\label{concrete2}
 P \ll P' \Longleftrightarrow \exists\; p \in P' \setminus P \mbox{ such that } P \subset I^-(p),
\end{equation}
and the inclusion of $\M$ into $\hat{\M}$ is $i:= i_{\M}: p \in \M \mapsto I^-(p) \in \hat{\M}$. In this context, we define the {\em future $c$-boundary} of $\M$ as
\[
\hat{\partial} \M := \hat{\M}\setminus i(\M).
\]
}

The chronological future/past in $\hat{\M}$ [resp. $\M$] will be denoted by $\hat{I}^{\pm}$ [resp. $I^{\pm}$]. We write $\hat{I}^{\pm}(U)$ for a subset $U \subset M$, instead of the more precise but heavier notation $\hat{I}^{\pm}(i_{\M}(U))$. In particular, $\hat{I}^{\pm}(p) = \hat{I}^{\pm}(\{p\}) = \hat{I}^{\pm}(i_{\M}(\{p\}))$.

For globally hyperbolic spacetimes with timelike boundary we have
\begin{proposition}
\label{achronalitybasic}
If $(\M,g)$ is a globally hyperbolic spacetime with timelike boundary, and $P,Q \in \hat{\M}$ are such that $P \ll Q$, then $P$ is a PIP. In particular, its future $c$-boundary $\hat{\partial} \M$ is achronal, i.e., there are no two points $P,Q \in \hat{\partial} \M$ such that $P \ll Q$.
\end{proposition}
\begin{proof}
Suppose $P \ll Q$, and let $q \in Q$ such that $P \subset I^-(q)$. If $P$ were a TIP, we could choose a future-inextendible timelike curve $\gamma:[a,b) \rightarrow \M$ such that $P = I^-(\gamma)$ (conf. Proposition \ref{PipST}). However, $\gamma$ would then be contained in the compact set $J^+(\gamma(a)) \cap J^-(q)$, which is impossible by Proposition \ref{imprisoned}. So $P$ is a PIP.
\end{proof}

The final result of this section establishes the relationship between the future $c$-completions of a globally hyperbolic with timelike boundary and its interior spacetime $(M,g)$, namely that they are in essence one and the same.

\begin{theorem}
\label{CLT1}
Let $(\M,g)$ be a globally hyperbolic with (possibly empty) timelike boundary. Then, the map $F: P \in \hat{\M} \mapsto P \cap M \in \hat{M}$ is a well-defined chronological bijection with chronological inverse. In particular, by Proposition \ref{ghlater}, the future completions of $(\M,g)$ and $(M,g)$ are isomorphic. Moreover, the image by $F^{-1}$ of the future $c$-boundary of $M$ is the (disjoint) union of the future $c$-boundary of $\M$ together with the smooth timelike boundary $\partial M$ of $\M$.
\end{theorem}
\begin{proof}
The proof proceeds in three steps.

\smallskip

\noindent {\em Step 1}. First, we show that $F$ is indeed well-defined, i.e., that given an IP $P$ in $\M$, $P \cap M$ is an IP in $M$. In order to do that, we use Proposition \ref{PipST} $(3)$. Let $(x _k)$ be a future-directed chain in $\M$ generating $P$. Now, for each $k \in \mathbb{N}$ we can choose $y_k \in I^-(x_{k+1},\M) \cap I^+(x_k,\M)\cap M$ (this follows easily from Remark \ref{nuevo}). Thus,
\begin{equation}
\label{mayuse}
I^{\pm}(p,\M) \cap M = I^{\pm}(p,M)
\end{equation}
for any $p \in M$, so we can conclude that $(y_k)$ is a future-directed chain in $(M,g)$, and it clearly generates $P \cap M$. Proposition \ref{PipST} $(3)$ now applied to $(M,g)$ now yields the desired result.

\smallskip

\noindent {\em Step 2}. We show that $F$ is a chronological bijection with chronological inverse. Consider past sets $P,Q \in \hat{\M}$ such that $P \ll Q$. Let $q \in Q$ such that $P \subset I^-(q,\M)$. Since $Q$ is open, we can pick $q'\in Q\cap I^+(q,\M)$. The open set $I^-(q',\M)\cap Q \cap I^+(q,\M)$ is non-empty and intersects $M$, so we can choose $q'' \in Q \cap I^+(q,\M)\cap M$. In particular, $P \subset  I^-(q'',\M)$, so $P\cap M \subset  I^-(q'',\M)$. But then Eq. (\ref{mayuse}) implies that $P\cap M \subset  I^-(q'',M)$, so that $P\cap M \ll_M Q\cap M$ and hence $F$ is chronological.

We now argue that $G: P \in \hat{M} \mapsto I^-(P,\M) \in \hat{\M}$ is well-defined, chronological, and that $G \equiv F^{-1}$. Now, given an IP $P$ in $M$, it is clear that any future-directed chain in $M$ generating $P$ will automatically be a future-directed chain in $\M$ generating $I^-(P,\M)$, so that the latter is an IP in $\M$. This shows that $G$ is indeed well-defined. To show it is chronological, let $P \ll_{M} Q$ be two IPs in $M$, and pick $q \in Q$ such that $P \subset I^-(q,M)$. This immediately implies that $I^-(P,\M) \subset I^-(q,\M)$, so that $G(P) \ll_{M} G(Q)$ as asserted. Using (\ref{mayuse}) it is straightforward to show that $G=F^{-1}$.

\smallskip

\noindent {\em Step 3}. A point on the future $c$-boundary of $M$ is a TIP $P$ in $(M, g)$. Let $\gamma:[a,b) \rightarrow M$ be future-inextendible timelike curve generating $P$ according to Proposition \ref{PipST} $(2)$. If $\gamma$ is future-inextendible in $(\M,g)$, then, as it also generates $F^{-1}(P) \equiv I^-(P,\M)$, the latter is also a TIP in $\M$, in which case it belongs to the future $c$-boundary of $\M$. If not, then $\gamma$ must have a future endpoint $p \in \partial M$. But in this case $F^{-1}(P) = I^-(P,\M) = I^-(p,\M) \equiv i_{\M}(p)$, so that $F^{-1}(P)$ can be identified with the point $p$ on $\partial M$. Conversely, given any TIP $P$ on $(\M,g)$, $F(P)$ is also a TIP in $(M,g)$. In addition, given a point $q \in \partial M$, $Q:= I^-(q,\M) \cap M = F(i_{\M}(q))$ is obviously a TIP in $(M, g)$, and hence an element of the future $c$-completion of the latter spacetime.
\end{proof}

\section{CLT I: Hausdorff closed limits}\label{top1}

Throughout this section, $X$ denotes a fixed {\em topological} space\footnote{We apologize to the reader for the slightly confusing use of $X$ for both a topological space (with certain extra properties) and for a generic chronological set. But we will soon mix these two structures together anyway, so we believe no great harm will be done!}. Our idea is to review the general facts about Hausdorff closed limits, and in particular the situation in which we can define the closed limit topology. We were unable to find this material described in the systematic way we need in the literature, so we have decided to give a somewhat detailed account here.

\begin{definition}
\label{upperlower}
Let $(A_n)$ be a sequence of subsets of $X$. The {\em Hausdorff upper and lower limits} of $(A_n)$ are defined, respectively, by

\begin{equation}\label{eq:infsup}
  \begin{array}{c}
  \limsup(A_n) := \{x \in X \, : \, \mbox{ each neighborhood of $x$ intersects infinitely many $A_n$'s}\},\\
\liminf(A_n) := \{x \in X \, : \, \mbox{ each neighborhood of $x$ intersects all but finitely many $A_n$'s}\}.
  \end{array}
\end{equation}

In particular, $\liminf(A_n) \subset \limsup(A_n)$. If $\limsup(A_n) = \liminf(A_n)$, then this common set is called the {\em Hausdorff closed limit} (of $(A_n)$).
\end{definition}

\begin{remark}
\label{subseq}
{\em Note that if $(A_n)$ is a sequence of subsets of $X$ and $(A_{n_k})$ is any subsequence, then
\[
\liminf(A_n) \subset \liminf(A_{n_k})\; \mbox{ and }\; \limsup(A_{n_k}) \subset \limsup(A_n).
\]
(Simple examples show that both inclusions are strict in general.) In particular, if the Hausdorff closed limit $A$ of $(A_n)$ exists, then the Hausdorff closed limit of $(A_{n_k})$ also exists and equals $A$. Moreover, we also have
\[
\limsup(A_n) = \limsup({\rm cl}(A_n)) \mbox{ and } \liminf(A_n) = \liminf({\rm cl}(A_n)).
\]}
\end{remark}

The following result has a straightforward proof left to the reader.
\begin{proposition}
\label{hausdlimprop}
Let $(A_n)$ be a sequence of subsets of $X$. Then:
\begin{enumerate}[label=(\roman*)]
\item The Hausdorff upper and lower limits $\limsup(A_n)$ and $\liminf(A_n)$ are closed sets. In particular, if the Hausdorff closed limit of $(A_n)$ exists, then it is a closed set.
\item Suppose $X$ is first-countable and all $A_n$ are non-empty. The Hausdorff closed limit of $(A_n)$ exists if and only if the Hausdorff upper and lower limits both coincide with the set
\[
A_{\infty} := \{ x \in X \, : \, \exists (x_n) \subset X \mbox{ with } x_n \in A_n \, ,  \forall n \in \mathbb{N} \mbox{ such that $x_n \rightarrow x$}\},
\]
so in particular $A_{\infty}$ coincides with the Hausdorff closed limit.
\end{enumerate}
\end{proposition}

Let us denote by $\mathbb{P}(X)$ the power set of $X$. We say that $\mathcal{F} \subset \mathbb{P}(X)$ is {\em $H$-closed} if for each sequence $(S_n) \subset \mathcal{F}$ with Hausdorff closed limit $S_{\infty}$, we have $S_{\infty} \in \mathcal{F}$.
The following result is a direct consequence of \cite[Section II-A]{FloresHausdorffseparabilityboundaries2016}:
\begin{proposition}
\label{finetop} The set
\[
\tau_H := \{ \mathcal{A} \subset  \mathbb{P}(X) \, : \, \mathbb{P}(X) \setminus \mathcal{A} \mbox{ is $H$-closed} \}
\]
is the finest topology in $\mathbb{P}(X)$ with the following property: \\
($\ast$) Given a sequence $(S_n)$ with Hausdorff closed limit $S_{\infty}$, $(S_n)$ converges to $S_{\infty}$.
\end{proposition}

\smallskip

From now on, we will refer to $\tau_H$ as {\em fine $H$-topology}.

We shall assume, for the rest of this section, that $X$ is metrizable, separable, and locally compact. We fix a topological metric $d$ on $X$. Given these topological assumptions on $X$, we can in addition assume, without loss of generality, that $d$ has the {\em Heine-Borel property}, i.e., that any closed $d$-bounded subset of $X$ is compact, and we do so in what follows. Using $d$, we wish to topologize the set
\[
\mathcal{C}_X:= \{ S \subseteq X \, : \, S \mbox{ is closed and non-empty}\}.
\]
Note that since the Hausdorff closed limit of any sequence of non-empty subsets of $X$ is closed and non-empty, $\mathcal{C}_X$ is a closed subset of $\mathbb{P}(X)$ in the fine $H$-topology. Denote by $C_c(X)$ the set of continuous real-valued functions on $X$ endowed with the topology of uniform convergence in compact subsets (which coincides with the compact-open topology). Since $C_c(X)$ is a well-known Frechet space, in particular it is metrizable. Consider then the map
\begin{equation}
\label{map}
\Phi: S \in \mathcal{C}_X \mapsto \Phi_S \in C_c(X)
\end{equation}
given by
\[
\Phi_S(x):= d(x,S),\; \forall x \in X,\; \forall S \in \mathcal{C}_X.
\]
The main properties of $\Phi$ are summarized in the following theorem. In its proof we shall repeatedly use the following simple technical result.
\begin{lemma}
\label{point}
Given $S \in \mathcal{C}_X$ and $x_0 \in X$, there exists $y_0 \in S$ such that
\[
d(x_0,S) = d(x_0,y_0).
\]
\end{lemma}
\begin{proof}
For each $n \in \mathbb{N}$, we can pick $y_n \in S$ such that
\begin{equation}
\label{lim}
d(x_0,S) \leq d(x_0, y_n) < d(x_0,S) + \frac{1}{n}.
\end{equation}
Thus, $d(x_0, y_n) \leq d(x_0,S) + 1$, i.e., the sequence $(y_n)$ is contained in the closed ball $\overline{B}_{d(x_0,S) + 1}(x_0)$ in $X$, which is compact by the Heine-Borel property of $d$. Therefore, up to passing to a subsequence, we may assume that $y_n \rightarrow y_0$ for some $y_0 \in S$ (since $S$ is closed). Continuity of $d$ gives that $d(x_0, y_n) \rightarrow d(x_0,y_0)$, and thus, taking the limit $n \rightarrow +\infty$ in (\ref{lim}) yields the result.
\end{proof}

\medskip

\begin{theorem}
\label{mapprop}
For the mapping $\Phi$, the following properties hold.
\begin{enumerate}[label=(\roman*)]
\item \label{mapprop-1} $\Phi$ is one-to-one. In particular, there exists a unique (metrizable) topology $\tau_c$ on $\mathcal{C}_X$ for which $\Phi$ is a homeomorphism onto its image (with the induced topology).
\item \label{mapprop-2} For each $S \in \mathcal{C}_X$, $\Phi_S$ is Lipschitz. Indeed,
\[
|\Phi_S(x) - \Phi_S(y)| \leq d(x,y),\quad\forall x,y \in X.
\]
In particular, the image $\Phi(\mathcal{C}_X)$ is an equicontinuous subset of $C_c(X)$.
\item \label{mapprop-3} If a sequence $(S_n) $ in $\mathcal{C}_X$ converges in the topology $\tau_c$ to $S \in \mathcal{C}_X$, then $S$ is the Hausdorff closed limit of $(S_n)$.
\item \label{mapprop-4}  $\Phi(\mathcal{C}_X)$ is closed in $C_c(X)$.
\item \label{mapprop-5} If a sequence $(S_n) \subset \mathcal{C}_X$ has a Hausdorff closed limit $S_{\infty} \in \mathcal{C}_X$, then $S_n \rightarrow S_{\infty}$ in $\tau_c$.
\end{enumerate}
\end{theorem}
\begin{proof}
\textit{\ref{mapprop-1}} Given $S,S' \in \mathcal{C}_X$ such that $\Phi_{S} = \Phi_{S'}$, we have
\[
x \in S \Longrightarrow \Phi_S(x) = d(x,S) = 0 \Longrightarrow d(x,S') = \Phi_{S'}(x) = 0 \Longrightarrow x \in S',
\]
that is, $S \subset S'$. Analogously, $S' \subset S$, and hence $S=S'$.

\smallskip

\textit{\ref{mapprop-2}} This is a standard result in metric spaces.

\smallskip

\textit{\ref{mapprop-3}} Consider a sequence $(S_n) \subset \mathcal{C}_X$ such that $S_n\rightarrow S$ in $\tau_c$, with $S \in \mathcal{C}_X$. Since $\Phi$ is continuous by the definition of $\tau_c$, $\Phi_{S_n} \rightarrow \Phi_S$ in $C_c(X)$. Let $x \in \limsup(S_n)$. Then there exists a subsequence $(S_{n_k})$ of $(S_n)$ and a sequence $(x_k)$ in $X$ such that
 \[
 x_k \in B_{1/k}(x)\cap S_{n_k}, \forall k \in \mathbb{N}.
 \]
 In particular, $x_k \rightarrow x$, and $(x_k)$ is thus bounded, and from the Heine-Borel property of $d$, may be viewed as contained inside a compact set $K$. Uniform convergence in $K$ now yields
\begin{equation}
\label{standard}
0 \equiv d(x_k,S_{n_k}) = \Phi_{S_{n_k}}(x_k) \rightarrow \Phi_S(x),
\end{equation}
whence we deduce that $d(x,S) \equiv \Phi_S(x) =0$, and then $x \in S$. We conclude that $\limsup (S_n) \subset S$.

Now, let $y \in S$. Then $\Phi_S(y) =0$, and hence $d(y,S_n) \rightarrow 0$. Using Lemma \ref{point}, pick $y_n \in S_n$ for each $n \in \mathbb{N}$ with $d(y,y_n) \equiv d(y,S_n)$. For any $r >0$, $d(y,y_n) <r$ for large enough $n$, i.e. $y_n\in S_n \cap B_r(y) \neq \emptyset$. We thus conclude that $y \in \liminf (S_n)$, and then that $S \subset \liminf (S_n)$. \\
\textit{\ref{mapprop-4}} Since  $C_c(X)$ is metrizable, to show that $\Phi(\mathcal{C}_X)$ is closed therein we only need to show that it is sequentially closed. Pick again a sequence $(S_n) \subset \mathcal{C}_X$ such that $\Phi_{S_n} \rightarrow \Theta$ in $C_c(X)$. Let
\[
S:= \Theta ^{-1}(0).
\]
We claim that $S \in \mathcal{C}_X$. Since $S$ is obviously closed, we only need to show that it is non-empty. Pick any $w_0 \in X$. We have $\Phi_{S_n}(w_0) \rightarrow \Theta(w_0)$ in $\mathbb{R}$ and hence there exists a number $a_{0} >0$ for which $|\Phi _{S_n}(w_0)| < a_{0}$ for all $ n \in \mathbb{N}$. Therefore, by Lemma \ref{point} we can choose a sequence $(w_n) \subset X$ such that $w_n \in S_n$ and $d(w_0,w_n) < a_{0}$ for all $ n \in \mathbb{N}$. In particular, this sequence is contained in the closed ball $\overline{B}_{a_{0}}(w_0)$, which is compact by the Heine-Borel property of $d$. Then, some subsequence $w_m$ converges in $X$ to some $z_0 \in \overline{B}_{a_{0}}(w_0)$, say. The uniform convergence of $\Phi_{S_n}$ in $\overline{B}_{a_{0}}(w_0)$ then implies that
\begin{equation}
\label{standard2}
0 \equiv d(w_m,S_m) = \Phi_{S_m}(w_m) \rightarrow \Theta(z_0),
\end{equation}
whence we conclude that $z_0 \in S$. Thus $S$ is non-empty.

We now wish to show that $\Theta \equiv \Phi_S$. By definition,
\[
\Theta(x) = 0 \Leftrightarrow x \in S \Leftrightarrow \Phi_S(x) =0,
\]
whence $\Phi_S$ and $\Theta$ coincide in $S$ (being identically zero there). Consider an arbitrary point $x_0 \in X$. First, use Lemma \ref{point} to choose $y_0 \in S$ such that $d(x_0,S) = d(x_0,y_0)$. In particular, $\Theta(y_0) \equiv 0$, and hence $\Phi_{S_n}(y_0) \rightarrow 0$. Using Lemma \ref{point} again, we choose $y_n \in S_n$ for each $n \in \mathbb{N}$ with $d(y_0,y_n) = d(y_0,S_n)$. We conclude that $y_n \rightarrow y_0$. Continuity of $d$ then gives
\[
d(x_0,y_n) \rightarrow d(x_0,S).
\]
Since $\Phi_{S_n}(x_0) \leq d(y_n, x_0)$ for all $n \in \mathbb{N}$, we conclude that $\Theta(x_0) \leq d(x_0,S) = \Phi_S(x_0)$. We can also pick a sequence $(z_n) \subset X$ such that $z_n \in S_n$ and
\begin{equation}
\label{specious}
d(x_0,z_n) < d(x_0,S_n) + \frac{1}{n} = \Phi_{S_n}(x_0) + \frac{1}{n} ,
\end{equation}
and thus, up to passing to a subsequence we have $z_n \rightarrow z_0$ for some $z_0 \in S$, by an argument entirely analogous to that around (\ref{standard2}). Taking $n \rightarrow +\infty$ in (\ref{specious}) we get
\[
\Phi _S(x_0) = d(x_0,S) \leq d(x_0,z_0) \leq \Theta(x_0).
\]
We conclude that $\Phi_S(x_0) = \Theta(x_0)$, and hence that $\Theta = \Phi_S \in \Phi (\mathcal{C}_X)$ as claimed.

\smallskip

\textit{\ref{mapprop-5}} Let $(S_n) \subset \mathcal{C}_X$ be a sequence with a Hausdorff closed limit $S_{\infty} \in \mathcal{C}_X$. Suppose, by way of contradiction, that $(S_n)$ does not converge in $\tau_c$ to $S_{\infty}$. Then there exist a $\tau_c$-open set $\mathcal{A} \in \tau_c$ containing $S_{\infty}$ and a subsequence $(S_m)_{m \in \mathbb{N}'}$ of $(S_n)$ contained in $\mathcal{C}_X \setminus \mathcal{A}$. Now, by Remark \ref{subseq}, $S_{\infty}$ is also the Hausdorff closed limit of $(S_m)$. Fix a $y_0 \in S_{\infty}$. By Proposition \ref{hausdlimprop} (ii), we can choose $y_m \in S_m$ for each $m \in \mathbb{N}'$ with $y_m \rightarrow y_0$.  Then, using item \textit{\ref{mapprop-1}} of the current theorem, we have
\[
|\Phi_{S_m}(y_m) (\equiv 0)- \Phi_{S_m}(y_0)| \leq d(y_m,y_0) \rightarrow 0,
\]
i.e., $\Phi_{S_m}(y_0) \rightarrow 0$. In particular, there is a number $c>0$ such that $|\Phi_{S_m}(y_0)| \leq c$, for $m \in \mathbb{N}'$. Given $x_0 \in X$ we again use \textit{\ref{mapprop-2}} to get
\[
|\Phi_{S_m}(x_0) - \Phi_{S_m}(y_0)| \leq d(x_0,y_0) \Longrightarrow |\Phi_{S_m}(x_0)| \leq c + d(x_0,y_0).
\]
Since $(\Phi_{S_m})$ is equicontinuous, we can use the Arzel\'{a}-Ascoli theorem to conclude that there exists a sub-subsequence $(\Phi_{S_k})$ converging to $\Theta \in C_c(X)$ uniformly in compact subsets. Since $\Phi(\mathcal{C}_X)$ is closed, $\Theta = \Phi_S$ for some $S \in \mathcal{C}_X$. By $(iii)$, $S \equiv S_{\infty}$, since $S_{\infty}$ is also the Hausdorff closed limit of $(S_k)$. Since $\Phi$ is a homeomorphism, we conclude that $S_k \rightarrow S_{\infty}$ in $\tau_c$. But then, for large enough $k$, $S_k \in \mathcal{A}$, a contradiction. Therefore, we must have that $S_n \rightarrow S_{\infty}$ in $\tau_c$.
\end{proof}

The final result of this section is a consequence of Theorem \ref{mapprop} (iii), (iv), and the fact that a sequential topology is determined by its limits  \cite[Lemma 2.4]{FloresHausdorffseparabilityboundaries2016}.
\begin{corollary}
\label{equivalence}
The topology $\tau_c$ coincides with the topology induced on $\mathcal{C}_X \subset \mathbb{P}(X)$ by $\tau_H$.
\end{corollary}
%

\begin{remark}
\label{nice}
{\em The relevance of Corollary \ref{equivalence} lies in that it shows that the topology $\tau_c$ actually does not depend on the particular choice of the metric $d$, provided $X$ is, as assumed, metrizable, separable and locally compact (e.g., if $X$ is a finite-dimensional smooth manifold with or without boundary).}
\end{remark}

\section{CLT II: a metrizable topology on the future $c$-completion of globally hyperbolic spacetimes with timelike boundary}\label{top2}

We return to the context of spacetimes, and shall assume, throughout this section, that $(\M,g)$ is globally hyperbolic with (possibly empty) timelike boundary. Our purpose here is to define a suitable metrizable topology on its future c-completion in this case. We emphasize again that we always adopt the ``concrete'' future c-completion $(\hat{\M}, \ll, i_{\M})$ given through the end of section \ref{prelim}.

In order to topologize $\hat{\M}$, we shall adopt the topology $\tau_c \equiv \tau_H$ on $\mathcal{C}_{\M}$ as described in the previous section\footnote{Of course, we now take the metric space $X$ used in the previous section to be $\M$ with an appropriate metric $d:\M \times \M \rightarrow \mathbb{R}$.}. As mentioned in the Introduction, Hausdorff closed limits of sequences of even very well-behaved sets do not in general preserve their good properties. Remarkably, however, the authors of \cite{GallowayHausdorffClosedLimits2017} show that sequences of {\em past sets} are well-behaved, insofar as Hausdorff closed limits of sequences of past sets are also (the closure of) past sets (see especially Section 3 of \cite{GallowayHausdorffClosedLimits2017}). Although the Hausdorff closed limit of a sequence of IPs need not be the closure of an IP, we still can obtain good properties for the CLT.

Recall that by Proposition \ref{p_ord_lowerlevels}, the interior $(M, g)$ is a causally continuous spacetime on its own right. Thus, by the general discussion in section \ref{prelim} and more specifically  Proposition \ref{causalityforpastdetermination}, it will have its own standard future completion, denoted by $(\hat{M}, \ll_M, i)$. We now have

\begin{theorem}
\label{CLT}
If $(\M,g)$ is a globally hyperbolic with (possibly empty) timelike boundary, then the following statements hold.
\begin{enumerate}[label=(\roman*)]
\item \label{CLT-2} The mapping $\chi: P \in \hat{\M} \mapsto {\rm cl}(P) \in \mathcal{C}_{\M}$ is one-to-one. Therefore, there exists a unique (metrizable) topology on $\hat{\M}$ for which $\chi$ is a homeomorphism onto its image. We shall refer to this topology as the {\em closed limit topology} (CLT), and denote it again by $\tau_c$.
\item \label{CLT-3} If $(P_k)$ be a sequence of IPs in $(\M,g)$, and $P$ is an IP therein, $P_k \rightarrow P$ with respect to the CLT if and only if $\overline{P}$ is the Hausdorff closed limit of $(P_k)$.
\end{enumerate}

\end{theorem}
\begin{proof}
\textit{\ref{CLT-2}}  We have to check that for any past sets $P,Q \in \hat{\M}$,
\[
{\rm cl}(P)={\rm cl}(Q)  \Longrightarrow P =Q.
\]
Now,
\[
{\rm cl}(P)= {\rm cl}(Q) \Longrightarrow {\rm cl}(P)\cap M = {\rm cl}(Q)\cap M \Longrightarrow {\rm cl}_M(P \cap M) = {\rm cl}_M(Q \cap M).
\]
By Proposition 3.2 in \cite{GallowayHausdorffClosedLimits2017}, it follows that $P \cap M = Q \cap M$. But then $F(P) = F(Q) \Longrightarrow P=Q$, where $F$ is defined as in Theorem \ref{CLT1}. \\
\textit{\ref{CLT-3}} is immediate from Corollary \ref{equivalence}.
\end{proof}

The closed limit topology introduced in Theorem \ref{CLT} (i) has the following important key properties, which show in a precise manner what we mean when we say that the CLT is ``good'' (compare with \cite[Theorem 3.27]{Floresfinaldefinitioncausal2011}).
\begin{theorem}
\label{thething}
If $(\M,g)$ is a globally hyperbolic with (possibly empty) timelike boundary, then the following facts hold for the CLT $\tau_c$ on $\hat{\M}$.
\begin{enumerate}[label=(\roman*)]
\item \label{thething-1} $i(\M)$ is dense in $\hat{\M}$. In particular, $\hat{\partial}\M$ has empty interior.
\item \label{thething-2} $i: \M \hookrightarrow \hat{\M}$ is an open continuous map. In particular, $i(\M)$ is an open (dense) set in $\hat{\M}$, the induced topology on $\M$ by $i$ coincides with the manifold topology, and $(\hat{\M},\tau_c)$ is second-countable.
\item \label{thething-3} $\hat{I}^{\pm}(P)$ is open in $\hat{\M}$, for all $P \in \hat{\M}$.
\item  \label{thething-4} Any future-directed chain $(P_n) \subset \hat{\M}$ converges in $\tau_c$.
\end{enumerate}
\end{theorem}
\begin{proof}
\textit{\ref{thething-1}} Given $P \in \hat{\M}$, let $(x_n)$ be a future-directed chain in $\M$ generating $P$ (cf. Prop. \ref{PipST}(3)). Then $P = \cup _{n \in \mathbb{N}} I^-(x_n)$. Our strategy is to show that $i(x_n) \equiv I^-(x_n) \rightarrow P$ in $\tau_c$. But by item (ii) in Theorem \ref{CLT},
\[
I^-(x_n) \rightarrow P  \Longleftrightarrow  \mbox{ ${\rm cl}(P)$ is the Hausdorff closed limit of $(I^-(x_n))$}.
\]
Now, let $x \in \limsup (I^-(x_n))$, and $U \ni x$ open. For some infinite set $\mathbb{N}' \subset \mathbb{N}$, $I^-(x_m)\cap U \neq \emptyset, \, \forall m \in \mathbb{N}'$, and so in particular $P \cap U \neq \emptyset$. We conclude that $x \in {\rm cl}(P)$. Hence, $\limsup (I^-(x_n)) \subset {\rm cl}(P)$. Let $ y \in {\rm cl}(P)$, and again take $V \ni y$ open. Pick $z \in P \cap V$. Then, for large enough $n \in \mathbb{N}$, $z \ll_g x_n$, and hence
\[
z \in I^-(x_n) \cap V,
\]
and we conclude that $y \in \liminf (I^-(x_n))$. Thus,
\[
{\rm cl}(P) \subset \liminf (I^-(x_n)) \subset \limsup (I^-(x_n)) \subset {\rm cl}(P),
\]
and $\overline{P}$ is the Hausdorff closed limit of $(I^-(x_n))$ as desired.

\smallskip

\textit{\ref{thething-2}} Let us prove first that $i$ is continuous. Let $(x_n)$ be a sequence in $\M$ converging to some $x_0 \in \M$. We wish to show that $i(x_n) \rightarrow i(x_0)$ in $\tau_c$. But again
\[
I^-(x_n) \rightarrow I^-(x_0) \Longleftrightarrow  \mbox{ ${\rm cl}(I^-(x_0))$ is the Hausdorff closed limit of $(I^-(x_n))$}.
\]
Now, let $x \in \limsup (I^-(x_n))$, and $U \ni x$ open. Using strong causality (cf. Prop. \ref{imprisoned}), we may find $x_{-},x_{+} \in U$ such that
\[
x \in I^+(x_{-}) \cap I^-(x_{+}) \subset J^+(x_{-}) \cap J^-(x_{+}) \subset U,
\]
and therefore some infinite set $\mathbb{N}' \subset \mathbb{N}$ for which $I^-(x_m)\cap I^+(x_{-}) \cap I^-(x_{+}) \neq \emptyset$ whenever $m \in \mathbb{N}'$. Choose $ y_m \in I^-(x_m)\cap I^+(x_{-}) \cap I^-(x_{+})$ for each $m \in \mathbb{N}'$. This defines a sequence $(y_m)$ in the set $J^+(x_{-}) \cap J^-(x_{+})$, which is compact by global hyperbolicity. We may then assume that $y_m \rightarrow y_0$ for some $y_0 \in J^+(x_{-}) \cap J^-(x_{+}) \subset U$. But $y_m \ll_g x_m$ by construction, and again global hyperbolicity implies (see Proposition \ref{p_ord_lowerlevels}) that $y_0 \leq_g x_0$. In other words, $y_0 \in U \cap{\rm cl}(I^-(x_0))$, and hence $U \cap I^-(x_0) \neq \emptyset$. We conclude that $x \in {\rm cl}(I^-(x_0))$, and thus that $\limsup (I^-(x_n)) \subset {\rm cl}(I^-(x_0))$.

Let $y \in {\rm cl}(I^-(x_0))$, and again take $V \ni y$ open.  Let $y' \in V \cap I^-(x_0)$. Since in particular $x_0 \in I^+(y')$ and the latter set is open, for large enough $n \in \mathbb{N}$, $x_n \in I^+(y')$ and then $y' \in I^-(x_n)$. In other words, $V \cap I^-(x_n)\neq \emptyset$ for all $n \in \mathbb{N}$ large enough. This implies that $y \in \liminf (I^-(x_n))$, and we then conclude that ${\rm cl}(I^-(x_0)) \subset \liminf (I^-(x_n))$ and thus that ${\rm cl}(I^-(x_0))$ is the Hausdorff closed limit of $(I^-(x_n))$ as desired.

We wish to prove now that $i$ is an open map. Let $U \subset \M$ be any open set, and pick a sequence $(P_n) \subset \hat{\M}\setminus i(U)$ converging in $\tau _c$ to some $P \in \hat{\M}$. We wish to show that $P \in \hat{\M}\setminus i(U)$, in which case one may conclude that $\hat{\M}\setminus i(U)$ is closed, i.e., that $i(U)$ is open, as claimed.

Suppose not; that is, assume that $P \in i(U)$, so that $P = I^-(p)$ for some $p \in U$. Strong causality (cf. Prop. \ref{imprisoned}) implies that $\exists x_{-},x_{+} \in U$ such that $p \in I^+(x_{-}) \cap I^-(x_{+}) \subset J^+(x_{-}) \cap J^-(x_{+}) \subset U$.

Again, $P_n \rightarrow P$ in $\tau_c$ means that $\overline{P}$ is the Hausdorff closed limit of $(P_n)$. Moreover, $p \in {\rm cl}(I^-(p)) \equiv J^-(p)$, since $(\M,g)$ is in particular causally simple (cf. Defn. \ref{higherladder} and Prop. \ref{p_ord_lowerlevels}). Therefore, for large enough $n \in \mathbb{N}$, $P_n \cap  I^+(x_{-}) \cap I^-(x_{+}) \neq \emptyset$, and we may as well assume this is always the case. For each $n$, then, choose $q_n \in P_n \cap  I^+(x_{-}) \cap I^-(x_{+})$ and a future-directed chain $(x^n_k)_{k \in \mathbb{N}}$ generating $P_n$ (Prop. \ref{PipST}(3)).

Now, let $n \in \mathbb{N}$. Either $P_n = I^-(p_n)$ with $p_n \notin U$, or else $P_n$ is a TIP and $(x^n_k)_{k \in \mathbb{N}}$ is future-endless (Prop. \ref{PipST}(4)), in which case its terms must eventually leave the compact set $J^+(x_{-}) \cap J^-(x_{+})$ if they ever enter it. In any case, one may pick some term $x^n_{k_n} \notin J^+(x_{-}) \cap J^-(x_{+})$ such that $q_n \ll_g x^n_{k_n}$. Therefore, we may pick $y_n \in \partial I^-(x_{+}) \cap I^+(x_{-}) \cap P_n $. We have thus built a sequence $(y_n)$ in the compact set $\partial I^-(x_{+}) \cap J^+(x_{-})$. Passing to a subsequence if necessary, we may assume that $y_n \rightarrow y_0 \in \partial I^-(x_{+})$. We also have $y_0 \in {\rm cl}(P) \equiv J^-(p)$. But then
\[
y_0 \leq_g p \ll_g x_{+} \Longrightarrow y_0 \in I^-(x_{+}),
\]
which is open and contradicts the previous fact that $y_0 \in \partial I^-(x_{+})$. Hence, $P \in \hat{\M}\setminus i(U)$.

\smallskip

\textit{\ref{thething-3}} Let $Q \in \hat{\M}$, and again take $(P^{\pm}_n) \subset \hat{\M}\setminus  \hat{I}^{\pm}(Q)$ converging in $\tau _c$ to some $P^{\pm} \in \hat{\M}$. We wish to show that $P^{\pm} \in \hat{\M}\setminus \hat{I}^{\pm}(Q)$, in which case one may conclude that $\hat{\M}\setminus \hat{I}^{\pm}(Q)$ is closed and hence that $\hat{I}^{\pm}(Q)$ is open.

Consider the past case first, and assume that $P^- \in \hat{I}^{-}(Q)$. Then for some $q \in Q\setminus P^-$, $P^- \subset I^-(q)$. Global hyperbolicity then implies that $P^-$ is a PIP, i.e., $P^- =I^-(p)$ for some $p \in \M$. Since $Q$ is open, we may assume without loss of generality that $p \ll_g q$, and hence that ${\rm cl}(P^-) = {\rm cl}(I^-(p))\equiv J^-(p) \subset I^-(q)$. We can now reason exactly as we did to show that $i$ is open, just using $Q$ instead of $U$, $q$ as $x_{+}$ and any $x_{-} \in P$ to arrive at a contradiction. Thus we must conclude that $P^- \in \hat{\M}\setminus \hat{I}^{-}(Q)$ as desired.

Now, assume $P^+ \in \hat{I}^{+}(Q)$. Then for some $p \in P^+ \setminus Q$, $Q \subset I^-(p)$. Global hyperbolicity again implies that $Q$ must be a PIP. i.e., $Q =I^-(q)$ for some $q \in P^+$. Pick any $p' \in P^+ \cap I^+(p)$. Since $p' \in {\rm cl}(P^+)$, which is the Hausdorff closed limit of $(P^+_n)$, eventually $P^+_n \cap I^+(p) \neq \emptyset$, and if we pick $y_n$ in the latter intersection, $Q \subset I^-(y_n)$, and hence $Q \ll P^+_n$, a contradiction. Again, $P^+ \in \hat{\M}\setminus \hat{I}^{+}(Q)$ and the proof is complete.

\smallskip

\textit{\ref{thething-4}} Let a future-directed chain $(P_n)$ in $\hat{\M}$ be given, and let $P:= \cup _{n\in \mathbb{N}}P_n$.
We claim, first, that $P$ is an IP in $(\M,g)$. It is clearly a (non-empty) past set therein. Now, pick for each $n \in \mathbb{N}$ a future-directed chain $(x^{(n)}_k)_{k \in \mathbb{N}}$ generating $P_n$. Since $(P_n)$ is future-directed we can also choose for each $n\in \mathbb{N}$ a point $p_n \in P_{n+1}\setminus P_n$ for which $P_n \subset I^-(p_n,\M)$, and hence also some $k_n \in \mathbb{N}$ in such a way that $P_n \subset I^-(x^{(n+1)}_{k_n},\M) \subset P_{n+1}$ and the resulting sequence $(k_n)_{n \in \mathbb{N}}$ of natural numbers is strictly increasing. Clearly the sequence $(x^{(n+1)}_{k_n}))_{n \in \mathbb{N}}$ is a future-directed chain which generates $P$. By item (3) of Proposition \ref{PipST}, $P$ is an IP as claimed.

Now, we need to show that ${\rm cl}(P)$ is the Hausdorff closed limit of $(P_n)$.
Again, let $x \in \limsup (P_n)$. Since any neighborhood of $x$ in $\M$ intersects some $P_n$ it will also intersect $P$; thus $x \in {\rm cl}(P)$, and this establishes that $\limsup(P_n) \subset {\rm cl}(P)$. Given $y \in {\rm cl}(P)$ and any open set $U \ni y$, let $y' \in U \cap P$, so that $y' \in U \cap P_{n_0}$ for some $n_0 \in \mathbb{N}$. But since $(P_n)$ is increasing this means that $y' \in U \cap P_{n}$ for every $n \in \mathbb{N}$ larger than $n_0$, that is, $U \cap P_{n} \neq \emptyset$ for $n$ large enough; in other words, $y \in \liminf(P_n)$ and hence
\[
\limsup(P_n) \subset {\rm cl}(P) \subset \liminf(P_n)
\]
as desired. Thus, $P_n \rightarrow P$ in the CLT.
\end{proof}

\medskip

Henceforth, for any globally hyperbolic spacetime $(\M,g)$ with (perhaps empty) timelike boundary, the {\em (standard) future $c$-completion} will always refer to $(\hat{\M}, \ll, i_{\M})$ with $\hat{\M}$ endowed with the CLT topology $\tau_c$.

\begin{remark}
  \label{def:2} {\em We have adopted in this section the usual definition of Hausdorff limits as introduced, e.g., in \cite{GallowayHausdorffClosedLimits2017}. In this approach, the superior and inferior limits are closed sets. However, the CLT topology can be reformulated by using \emph{open} versions for the superior and inferior limits. This, in turn, is more convenient in comparing the CLT and the chronological topology (see the Appendix). Moreover, it provides an alternative proof of Theorem \ref{thething} by using \cite[Theorem 3.27]{Floresfinaldefinitioncausal2011} (cf. Remark A.2 in the Appendix).}
\end{remark}




\section{Conformal extensions and the CLT}\label{conf}

{\em Throughout this section we shall only deal with globally hyperbolic spacetimes without boundary (hence $\M\equiv M$).} As mentioned in the Introduction, our goal now is to compare the natural topology of the conformal boundary (in a suitably defined conformal extension) with the $c$-boundary $\hat{\partial}M := \hat{M}\setminus M$ of the standard future $c$-completion with the closed limit topology. We shall see that in this case conformal and $c$-boundaries are essentially ``the same'' given a few natural assumptions on the conformal extension. The importance of this coincidence is twofold. On one hand, it ensures the uniqueness of the conformal boundary independently of the (suitably defined) conformal extension. On the other hand, it establishes good regularity properties for the $c$-boundary.

The details in the definition of a conformal boundary for spacetimes change somewhat throughout the literature, so we have opted in this paper for a version which is weak enough to cover most of the best known cases.

\begin{definition}
\label{def3}
A {\em conformal extension} of $(M,g)$ is a triple $(\tilde{M},\tilde{g}, \varphi)$, where $(\tilde{M},\tilde{g})$ is a spacetime and $\varphi: M \rightarrow \tilde{M}$ is an open, time-orientation-preserving smooth embedding satisfying
\[
\label{conformalfactor1}
\varphi ^{\ast} \tilde{g} = \Omega^2 g
\]
for a a strictly positive function $\Omega \in C^{\infty}(M)$, the {\em conformal factor} of the conformal extension. The conformal extension $(\tilde{M},\tilde{g}, \varphi)$ is {\em nontrivial} if $\varphi(M) \neq \tilde{M}$, in which case the {\em future [resp. past] conformal boundary} associated with the extension $(\tilde{M},\tilde{g}, \varphi)$ is $\partial ^{+}M:= \partial \varphi(M) \cap I^+(\varphi(M),\tilde{M})$ [resp. $\partial ^{-}M:= \partial \varphi(M) \cap I^-(\varphi(M),\tilde{M})$.
\end{definition}

\begin{remark}\label{rmk1}
{\em A few comments about Definition \ref{def3} are in order.
\begin{enumerate}[label=(\arabic*)]
\item Note that we do not require, in this definition, that the past and/or future conformal boundaries have any regularity, and even if any one of them is smooth, we do not demand that the conformal factor extends smoothly to the boundary. As it stands, the definition is too weak to be useful in most applications, and has to be supplemented according to specific needs. Accordingly, we too shall make some extra assumptions in what follows, but they will still comprise most concrete globally hyperbolic examples in the literature.
\item In particular, if $(M,g)$ is extendible as a globally hyperbolic spacetime, i.e., it is {\em isometric} to a proper open subset of a larger globally hyperbolic spacetime $(\tilde{M}, \tilde{g})$, then the latter gives a conformal extension of the  the former with conformal factor $\Omega \equiv 1$. To avoid this kind of triviality, one usually assumes, in concrete applications, that $(M,g)$ is maximal in the globally hyperbolic class.
\item The ``larger'' spacetime used in a conformal extension may well be $(M,g)$ itself. As a standard example, consider the following. Let $M = \mathbb{R}^2$ with the flat metric
\[
g = -dudv,
\]
and the time orientation such that both $\partial_u$ and $\partial_v$ are future-directed null vector fields. This spacetime is geodesically complete, and hence inextendible. Define
\[
\varphi: (u,v) \in \mathbb{R}^2 \mapsto (\arctan u,\arctan v) \in \mathbb{R}^2
\]
and consider the smooth function
\[
\label{conformalfactor2}
\Omega : (u,v) \in \mathbb{R}^2 \mapsto \cos u \cos v \in \mathbb{R}.
\]
Then it is easy to check that $(M,g, \varphi)$ is a (nontrivial) conformal extension of $(M,g)$ thus given, such that the image of $M$ by $\varphi$ is the open square $Q:= (-\pi/2, \pi/2)^2$, with conformal factor $\Omega \circ \varphi$. Note that $\Omega$ vanishes on the boundary of $Q$ in $\mathbb{R}^2$. Thus, $(M,g)$ is conformally extended via a mapping onto an open set of itself.
\item Let $(\tilde{M},\tilde{g}, \varphi)$ be a conformal extension of $(M,g)$. Since $\varphi$ is a diffeomorphism when viewed as a map onto its open image $\varphi(M)$, consider its (smooth) inverse $\varphi^{-1}: \varphi(M) \rightarrow M$. Then $(\varphi(M),(\varphi ^{-1})^{\ast} g)$ is a spacetime on its own right, and moreover it is isometric to $(M,g)$ by construction. Hence, there is no loss of generality in regarding $M \subseteq \tilde{M}$ and $\varphi$ to be the inclusion map. We shall often do so in what follows, and will then abuse notation by referring to $(\tilde{M},\tilde{g})$ (without appealing to $\varphi$) as a conformal extension. In this case, we write the condition in Defn. \ref{conformalfactor1} as
\[
\label{conformalfactor3}
\tilde{g}|_{M} = \Omega^2 g.
\]
\item Even if we identify isometric spacetimes as in the previous item, conformal extensions, and also the associated conformal boundaries, are clearly far from being unique. However, (and this is one of the main points of this section) natural, sufficient additional conditions can be given on a conformal extension (cf. Definition \ref{nested} and Theorem \ref{causaltoconformal} below) that ensure that the conformal boundary is indeed uniquely defined.
\item The existence of a conformal extension, as well as the conformal extension itself, are conformally invariant notions in the following sense. If $(\tilde{M},\tilde{g})$ is a conformal extension of $(M,g)$ with conformal factor $\Omega \in C^{\infty}(M)$, then it is also a conformal extension (with the same conformal boundary) of the spacetime $(M,\omega^2 g)$, where $\omega \in C^{\infty}(M)$ is a strictly positive function, with conformal factor $\Omega/\omega$. Note, however, and this will be important in the next section, that {\em the extendibility of the conformal factor to the boundary, and the values that such an extension might have there, are NOT a conformally invariant notion, depending rather sensitively on how $\Omega$ and $\omega$ are related.}
\end{enumerate}
}
\end{remark}

In passing, we note that we can say something interesting about spacetimes which acquire a timelike conformal boundary for some conformal extension. Indeed, Theorem \ref{CLT1} immediately yields the following result:
\begin{corollary}
\label{timelikeisquick}
Consider a conformal extension $(\tilde{M},\tilde{g})$ of a (necessarily causally continuous) spacetime $(M,g)$ without boundary such that
\begin{enumerate}[label=(\arabic*)]
\item the future conformal boundary $\partial^+ M$ of $M$ is a smooth timelike hypersurface in $(\tilde{M},\tilde{g})$ and
\item $(M\cup \partial^+ M,\tilde{g}|_{M\cup \partial^+ M})$ is a globally hyperbolic spacetime with timelike boundary $\partial^+ M$ and interior $M$.
\end{enumerate}
Then $\partial^+ M$ is contained in the future $c$-boundary $\hat{\partial}M$ of $(M,g)$.
\end{corollary} \qcd

We now turn to specific conditions that ensure equality of conformal and $c$-boundaries. Let us first introduce a few preliminary definitions and results.

\begin{definition}
\label{def1}
Let $A \subset M$ be any set.
\begin{enumerate}[label=(\alph*)]
\item $A$ is {\em causally convex} if any causal curve segment in $(M,g)$ with endpoints in $A$ is entirely contained in $A$,
\item $A$ is {\em future-precompact} [resp. {\em past-precompact} ] is there exists a compact set $K \subset M$ such that $A \subset I^-(K)$ [resp. $A \subset I^+(K)$].
\end{enumerate}
The {\em future boundary} [resp. {\em past boundary}] of $A$ is
\[
\partial ^+ A := I^+(A) \cap {\rm Fr} A \mbox{ [resp. $\partial ^- A := I^-(A) \cap {\rm Fr} A$]},
\]
where ${\rm Fr} A$ denotes the topological boundary of $A$ in $M$.
\end{definition}

Note that any PIP [resp. PIF] is future-precompact [resp. past-precompact] in $(M,g)$.
Moreover, in the case when $(M,g)$ is globally hyperbolic and a set $A \subset M$ is both future- and past-precompact, then its is indeed precompact, which motivates the use of these terms.

\begin{proposition}
\label{convexitypastfuture}
Let $A\subset M$. The following statements hold.
\begin{enumerate}[label=(\roman*)]
\item \label{convexitypast-1} Suppose $A$ is open. Then, $\partial ^+ A = \emptyset$ [resp. $\partial ^- A = \emptyset$] if and only if $A$ is a future set [resp. a past set].
\item \label{convexitypast-2} If $A$ is causally convex, then $\partial ^{\pm} A$, if non-empty, are achronal $C^0$ hypersurfaces in $(M,g)$. If $A$ is also open, then $A \cup \partial ^{\pm} A$ is causally convex.
\item  \label{convexitypast-3} If $A$ is either a future or a past set in $(M,g)$, then $A$ is causally convex.
\end{enumerate}
\end{proposition}
\begin{proof}
\textit{\ref{convexitypast-1}} The ``if'' part is immediate. For the converse, assume that $\partial ^+ A = \emptyset$, and let $p \in A, q \in I^+(p)$. Pick any future-directed timelike curve $\alpha:[0,1] \rightarrow M$ such that $\alpha(0) =p$ and $\alpha(1) =q$. If $\alpha$ ever leaves $A$, then it must intersect $\partial ^+ A$, a contradiction. Thus $q \in A$, and we conclude that $A$ is a future set. Reasoning time-dually, we establish that $A$ is a past set when $\partial ^- A = \emptyset$.

\smallskip

\textit{\ref{convexitypast-2}} We give the proof for the future boundary only, since the past case follows by time-duality. Suppose $\partial ^{+} A$ is not achronal, and let $p, q \in \partial ^{+} A$ such that $p\ll_g q$. Thus, we can pick open sets $U,V \subset I^+(A)$ such that $p \in U$, $q \in V$ and
\[
p' \in U, q' \in V \Longrightarrow p' \ll_g q'.
\]
Since $p,q \in {\rm Fr} A$ we can choose $p' \in U \cap (M \setminus A)$ and $q' \in V \cap A$. We can also choose $p'' \in A \cap I^-(p')$. But then $p' \in I^+(p'') \cap I^-(q')\cap (M \setminus A)$, and therefore $A$ cannot be causally convex, contrary to our assumption.

To show that $\partial ^{+} A$ is a $C^0$ hypersurface, since it is achronal we only need to show that $\partial ^{+}A \cap edge(\partial ^+ A) = \emptyset$. Let $p \in \partial ^{+}A$. Then, in particular $p \in I^+(q)$ for some $q \in A$. Pick any future-directed timelike curve $\gamma: [0,1] \rightarrow M$ such that $\gamma(0) \in I^+(q)\cap I^-(p)$ and $\gamma(1) \in I^+(p)$. In particular, $\gamma[0,1] \subset I^+(A)$. The achronality of $\partial ^{+} A$ implies that $\gamma(1) \notin A$. We claim that $\gamma(0) \in A$. If not, noting that $I^+(\gamma(0)) \cap A \neq \emptyset$ because $p \in I^+(\gamma(0)) \cap {\rm Fr} A$, we could then juxtapose a future-directed timelike curve from $q \in A$ to $\gamma(0)$ with another future-directed timelike curve from $\gamma(0)$ to some point in $A$, thereby obtaining a future-directed timelike curve with endpoints in $A$ but not entirely containing therein. This in turn would contradict the fact that $A$ is causally convex. We conclude that $\gamma$ intersects $\partial ^{+}A$, and therefore $p \notin edge(\partial ^{+} A)$.

Now assume $A$ is open and $A \cup \partial ^{+} A$ is not causally convex. Consider (say) a future-directed causal curve segment $\alpha:[0,1] \rightarrow M$ with $\alpha(0),\alpha(1) \in A\cup \partial ^{+} A$ and $\alpha(t_0) \neq A\cup \partial ^{+} A$ for some $t_0 \in (0,1)$. Since $A$ is open, we necessarily have $\alpha(0),\alpha(1) \in I^+(A)$. So we can pick $p \in A$ such that $p \ll_g \alpha(0)$. Thus $p\ll_g \alpha(t_0)$. But any future-directed timelike curve from $p$ to $\alpha(t_0)$ must intersect $\partial ^{+} A$, and hence $p \ll_g r \ll_g \alpha(t_0)$ for some $r \in \partial ^+A$. Now, either $\alpha(1) \in \partial ^+ A$ or else $\alpha(1) \in A$, but in any case $\alpha|_{(t_0,1)}$ must intersect $\partial ^+ A$. Thus, $\alpha(t_0) < q$ for some $ q \in \partial ^+ A$, and hence $r \ll_g q$, contradicting the achronality of $\partial ^+ A$.

\smallskip

\textit{\ref{convexitypast-3}} If $A$ is a past [resp. future] set, then ${\rm Fr} A \equiv \partial ^+A$ [resp. ${\rm Fr} A \equiv \partial ^-A$], and ${\rm Fr} A$ is known to be achronal by standard results in Causal Theory. Since $A$ is open, the fact that $A$ is causally convex now follows from \textit{\ref{convexitypast-2}}.
\end{proof}

\begin{proposition}
\label{convexitygh1}
Assume that $(M,g)$ is globally hyperbolic, and let $U \subset M$ be an open, connected, globally hyperbolic subset. Then, the following statements are equivalent.
\begin{enumerate}[label=(\roman*)]
\item \label{item:conv1} $U$ is causally convex.
\item \label{item:conv2} Any achronal set in $(U,g|_U)$ is achronal in $(M,g)$.
\item\label{item:conv3} Any Cauchy hypersurface $S$ in $(U,g|_U)$ is achronal in $(M,g)$.
\end{enumerate}
\end{proposition}
\begin{proof}
$\ref{item:conv1} \Rightarrow \ref{item:conv2} \Rightarrow \ref{item:conv3}$ are immediate. Therefore, we only need to prove $\ref{item:conv3} \Rightarrow \ref{item:conv1}$. Let $\alpha:[0,1] \rightarrow M$ be a future-directed causal curve with endpoints in $U$. Since $U$ is open, there will be no loss of generality in our argument if we assume that $\alpha$ is timelike. Let $S \subset U$ be a Cauchy hypersurface in $(U,g|_U)$ chosen so that  $\alpha(0) \in I^+(S,U)$. Hence, $\alpha(1) \in I^+(S,U)$ as well, for $\alpha(1) \in S \cup I^-(S,U)$ would mean a violation of the achronality of $S$ in $(M,g)$. Suppose $\alpha(t_0) \notin U$ for some $t_0 \in (0,1)$. Also, write
\begin{eqnarray}
s_0 &:=& \inf \{ s \in [0,1] \, : \, \alpha[s,1] \subset U \}, \\
r_0 &:=& \sup \{ t \in [0,1] \, : \, \alpha[0,t] \subset U \}.
\end{eqnarray}
Then the restrictions $\alpha_+:= \alpha|_{[0,r_0)}$ and $\alpha_{-}:= \alpha|_{(s_0,1]}$ are future- and past-inextendible in $(U,g|_U)$, respectively, and $r_0 < t_0 < s_0$. But then $\alpha_{-}$ must intersect $S$ at some $s_0< s'$, say, so
\[
p \ll_g \alpha(0) \ll_g \alpha(s')
\]
for some $p \in S$, which violates the achronality of $S$ in $(M,g)$. Hence, $\alpha[0,1] \subset U$.
\end{proof}

\begin{theorem}
\label{convexitygh2}
Assume that $(M,g)$ is globally hyperbolic, and let $U \subset M$ be an (non-empty) open, connected, causally convex, future-precompact and globally hyperbolic set. Then $\partial ^+ U$ is homeomorphic to a Cauchy hypersurface $S \subset U$ in $(U,g|_U)$.
\end{theorem}
\begin{proof}
$S$ is a connected $C^0$ hypersurface in $M$. By Proposition \ref{convexitygh1}, it is achronal in $(M,g)$. Let $K \subset M$ be some compact set such that $U \subset I^-(K)$. In particular, $U$ cannot be a future set, and hence $\partial ^+ U$ is also a non-empty achronal $C^0$ hypersurface by Proposition \ref{convexitypastfuture} \textit{\ref{convexitypast-1}}-\textit{\ref{convexitypast-2}}. Let $X:M \rightarrow TM$ be a complete future-directed timelike smooth vector field in $(M,g)$, and denote its flow by $\phi$. We shall define a continuous mapping $\varphi : S \rightarrow \partial ^+ U$ as follows. Given $p \in S$, the future-inextendible timelike curve
\[
t \in [0, +\infty) \mapsto \phi_t(p)
\]
staring at $p$ must leave the compact set $J^+(p)\cap J^-(K)$, and hence must leave $U$. When its does, it will intersect $ \partial ^+ U$, and since this set is achronal, will do so at a unique point, which we define as $\varphi(p)$. The mapping thus defined is obviously one-to-one, and hence open by Invariance of Domain. Therefore, in other to show that $\varphi$ is the desired homeomorphism we only need to show that it is onto. Accordingly, let $q \in \partial ^+ U$, and consider the inextendible timelike curve
\[
\gamma : s \in (-\infty, +\infty) \mapsto \phi_s(q).
\]
Then, for some $\varepsilon >0$, $\gamma[-\varepsilon,0] \subset I^+(U)$.

We claim that $\gamma(-\varepsilon) \in U$. If not, noting that $I^+(\gamma(-\varepsilon)) \cap U \neq \emptyset$ because $q \in I^+(\gamma(-\varepsilon)) \cap \partial U$, we could then juxtapose a future-directed timelike curve from some $r \in U$ to $\gamma(-\varepsilon)$ with another future-directed timelike curve from $\gamma(\varepsilon)$ to some point in $U$, thereby obtaining a future-directed timelike curve with endpoints in $U$ but not entirely containing therein. This in turn would contradict the fact that $U$ is causally convex. We conclude that $\gamma$ intersects $U$, and hence must intersect $S$. If this intersection occurred at the future of $q$ this would violate the achronality of $\partial ^+ U$, since we saw above that the future of any point in $S$ intersects $\partial ^+ U$. Therefore, for some $s_0 < 0$, $\phi_{s_0}(q) \in S$. But then we of course have $\varphi (\phi_{s_0}(q)) \equiv q$, thus concluding the proof.
\end{proof}

\medskip

In the spirit of improving Definition \ref{def3} for concrete uses, we use the following notions, inspired by the approach of \cite{OlafConformal}.
\begin{definition}
\label{nested}
A conformal extension $(\tilde{M},\tilde{g})$ of a globally hyperbolic spacetime $(M,g)$ is {\em future-nesting} if
\begin{itemize}
\item[N1)] $(\tilde{M},\tilde{g})$ is also globally hyperbolic,
\item[N2)] $M \subset \tilde{M}$ is causally convex and future-precompact.
\end{itemize}
\end{definition}

Proposition \ref{convexitypastfuture} and Theorem \ref{convexitygh2} immediately give the following regularity result.
\begin{corollary}
\label{convexitygh3}
If $(\tilde{M},\tilde{g})$ is a future-nesting conformal extension of a globally hyperbolic spacetime $(M,g)$, then the corresponding future conformal boundary $\partial^+ M$ is an achronal $C^0$ hypersurface in $(\tilde{M},\tilde{g})$ homeomorphic to a Cauchy hypersurface in $(M,g)$.
\end{corollary}\qcd

\begin{remark}\label{rmk2}
{\em The standard conformal extensions, in the Einstein cylinder, of Minkowski spacetime, past-incomplete Friedmann-Robertson-Walker models with non-negative cosmological constant, and de Sitter spacetime are future-nesting in this sense. The standard conformal extension of the Schwarzschild-Kruskal plane is of course also future-nesting.}

\end{remark}

We are finally ready to state and prove the main result of this section. It essentially means that for future-nesting conformal extensions, the future conformal boundary coincides with the future $c$-boundary when we adopt the CLT in the latter. Moreover, the conformal (and hence the $c$-)boundary is large enough to contain all the future endpoints of future-inextendible null geodesics in this case.

We shall need a technical lemma.
\begin{lemma}
\label{technical}
Let $(\tilde{M},\tilde{g})$ be a conformal extension of a spacetime $(M,g)$ for which $M \subset \tilde{M}$ is causally convex. Given $p \in \partial ^+ M$, $q \in I^-(p, \tilde{M})$ and any future-directed timelike curve $\alpha:[0,1] \rightarrow \tilde{M}$ from $q$ to $p$, there exists some $0<\varepsilon <1$ for which $\alpha[1-\varepsilon,1) \subset M$.
\end{lemma}
\begin{proof}
Since $p \in I^+(M,\tilde{M})$ and the latter set is open, for some $0< \varepsilon <1$, $\alpha[1- \varepsilon,1] \subset I^+(M,\tilde{M})$ by continuity. Let $t_0 \in [1-\varepsilon,1)$. Since $p \in I^+(\alpha(t_0),\tilde{M})$, $I^+(\alpha(t_0),\tilde{M} )\cap M \neq \emptyset$ since $p$ is a boundary point. Choose a future-directed timelike curve segment $\beta:[0,1] \rightarrow \tilde{M}$ starting at $\alpha(t_0)$ and ending on $M$. Thus, if $\alpha(t_0) \notin M$, we could choose a timelike curve segment in $\tilde{M}$ with endpoints in $M$ which leaves $M$, violating the causal convexity of $M$. Therefore, $\alpha(t_0) \in M$, and we then conclude that $\alpha[1-\varepsilon,1) \subset M$. \end{proof}

We are ready to state and prove the main result on this section.

\begin{theorem}
\label{causaltoconformal}
Let $(\tilde{M},\tilde{g})$ be a future-nesting conformal extension of a globally hyperbolic spacetime $(M,g)$. Then
\begin{enumerate}[label=(\arabic*)]
\item $(M \cup \partial^+ M ,\ll_{\tilde{g}},\varphi)$ is a future completion of $(M,g)$ (in the sense of section \ref{prelim}), where $\varphi :M \hookrightarrow M \cup \partial^+ M$ is (induced by) the inclusion of $M$ into $\tilde{M}$.
\item $M\cup \partial^+ M \subset \tilde{M}$ with the induced topology is homeomorphic to the set $\hat{M}$ of non-empty IPs in $(M,g)$ with the CLT topology $\tau_c$. The homeomorphism can be chosen so that it maps $\partial ^+ M$ onto $\hat{\partial} M$ and coincides with the inclusion $i: p \in M \mapsto I^-(p) \in \hat{M}$ on $M$.
\item Every null geodesic in $(M,g)$ which is future-inextendible therein has a future endpoint on $\partial ^+M$.
\end{enumerate}
\end{theorem}
\begin{proof}  Consider the map
\[
\Psi : p \in M\cup \partial^+ M \mapsto I^-(p,\tilde{M}) \cap M \in \hat{M}.
\]
Note that indeed $\Psi(p) \equiv I^-(p)$ for every $p \in M$ since $M$ is causally convex.

\smallskip

\noindent {\em Claim 1:} $\Psi$ is well-defined.\\
Let $p \in M\cup \partial^+ M$ be given. Now, $I^-(p,\tilde{M}) \cap M$ is clearly non-empty and open in $M$, so we need to show it is an IP. Let $q \in I^-(p,\tilde{M}) \cap M$. Pick $q'\in M$ such that $q'\ll_g q$. Thus, $q'\ll_{\tilde{g}} p$, and hence $q'\in I^-(p,\tilde{M}) \cap M$, so the latter set is a past set in $(M,g)$. Now, let $x,y \in I^-(p,\tilde{M}) \cap M$. Since $I^-(p,\tilde{M})$ is directed in $(\tilde{M},\ll_{\tilde{g}})$, we can find $z \in I^-(p,\tilde{M})$ such that $x,y \ll_{\tilde{g}} z$. Since $M$ is in particular open and causally convex in $\tilde{M}$, $M \cup \partial^+ M$ is causally convex by Proposition \ref{convexitypastfuture} (ii), so that $z \in M \cup \partial^+ M$. Now, either $p \in M$, so that $z\in M$ from the causal convexity of $M$, or else $p \in \partial^+ M$, and in this case the achronality of the future boundary precludes $z \in \partial^+ M$, so $z \in M$ in any case. Again, since $M$ is causally convex we have $x,y \ll_{g} z$. Therefore, we conclude that $I^-(p,\tilde{M}) \cap M$ is directed in $(M,g)$, and hence an IP (cf. Prop. \ref{PipST}(3)). Thus, $I^-(p,\tilde{M}) \cap M \in \hat{M}$.

\smallskip

$(1)$ We now wish to show that $\Psi$ is a chronological bijection with a chronological inverse, and apply Proposition \ref{ghlater}.

\smallskip

\noindent {\em Claim 2:} $\Psi$ is one-to-one.\\
Suppose $I^-(p,\tilde{M}) \cap M =I^-(q,\tilde{M}) \cap M$, for $p,q \in M\cup\partial^+ M$. Given $r \in I^-(p,\tilde{M})$, let $\alpha:[0,1] \rightarrow M$ be a future-directed timelike curve segment from $r$ to $p$. We claim that $\alpha$ intersects $M$. If $p \in M$ this is trivial, so we can assume $p \in \partial^+ M$, in which case the claim follows from Lemma \ref{technical}. Thus, $r \ll_{\tilde{g}} r'$ for some $r'\in I^-(p,\tilde{M}) \cap M =I^-(q,\tilde{M}) \cap M$. We conclude that $r \in I^-(q,\tilde{M})$, and hence that $I^-(p,\tilde{M}) \subset I^-(q,\tilde{M})$. We analogously show that $I^-(q,\tilde{M}) \subset I^-(p,\tilde{M})$, and hence $I^-(p,\tilde{M}) = I^-(q,\tilde{M})$. Since $(\tilde{M},\tilde{g})$ is strongly causal, it is in particular past-distinguishing, and thus $p=q$.

\smallskip

\noindent{\em Claim 3:} $\Psi$ is onto.\\
Let $P$ be a non-empty IP in $M$. Suppose first $P$ is a PIP, and write $P = I^-(p,M)$, for some $p \in M$. Clearly, $I^-(p,M) \subset I^-(p,\tilde{M}) \cap M$, and since $M$ is causally convex, $I^-(p,\tilde{M}) \cap M \subset I^-(p,M)$, and hence $P = \Psi(p)$. Note that this also establishes that
\[
\Psi \circ \varphi = i,
\]
where $i:M \hookrightarrow \hat{M}$ is the standard inclusion.

Now, assume $P$ is a TIP in $(M,g)$, and let $(x_n)$ be a future-endless chain generating $P$. Since $(\tilde{M},\tilde{g})$ is a future-nesting conformal extension, we may pick a compact set  $K \subset \tilde{M}$ such that $M \subset J^-(K, \tilde{M})$, and therefore $(x_n)$ is contained in the compact set $J^+(x_1,\tilde{M})\cap J^-(K, \tilde{M})$. Passing to a subsequence if necessary, we may then assume that $x_n \rightarrow p$ for some $p \in \partial M$. Since $x_1 <_{\tilde{g}} p$, and $q \ll_g x_1$ for some $q \in P \subset M$, $p \in I^+(M, \tilde{M})$. Therefore, $p \in \partial^+ M$. However, $(x_n)$ is also a future-chain in the strongly causal spacetime $(\tilde{M},\tilde{g})$. By Proposition \ref{PipST} (4), we conclude that $p$ is the future limit in $\tilde{M}$ of this chain. In particular, $x_n \ll_{\tilde{g}} p$ for all $n \in \mathbb{N}$. Therefore,
\[
P  \subset I^-(p,\tilde{M}) \cap M .
\]
Again, the fact that $p$ is a future limit means that for $r \in I^-(p,\tilde{M}) \cap M$ we must have $r \ll_{\tilde{g}} x_n$ for large enough $n$, and thus $r \ll_g x_n$ for such $n$, since $M$ is causally convex. Therefore $r \in P$ since the latter is a past set. We again conclude that $P = \Psi(p)$.\\

We have now shown that $\Psi$ is a bijection. It is clear that $\Psi$ is chronological. To show that $\Psi ^{-1}$ is also chronological, take $P,Q \in \hat{M}$ with $P \ll Q$. Let $p,q \in M \cup \partial ^+ M$ such that
\[
P= I^-(p,\tilde{M}) \cap M \mbox{ and } Q = I^-(q,\tilde{M}) \cap M.
\]
For some $x \in Q$, $P \subset I^-(x,M)$, and therefore $p \in {\rm cl}_{\tilde{M}}(I^-(x,M)) \equiv J^-(x, \tilde{M})$. Thus
\[
p \leq_{\tilde{g}} x \ll_{\tilde{g}} q \Longrightarrow p \ll_{\tilde{g}} q,
\]
thus showing that $\Psi ^{-1}(P) \ll_{\tilde{g}} \Psi ^{-1}(Q)$ as desired. This concludes the proof of $(1)$.

\smallskip

$(2)$ We wish to show that with the indicated topologies, $\Psi$ is a homeomorphism. First we note the following.

\smallskip

\noindent {\em Claim 4.} $\Psi(\partial ^+ M) = \hat{\partial} M$. \\
This equality is tantamount, for $p \in M \cup \partial ^+ M$, to the equivalence
\[
p \in \partial ^+ M \Longleftrightarrow I^-(p,\tilde{M}) \cap M \mbox{ is a TIP in $M$}.
\]
The ``only if'' part is clear, since if $p \in \partial ^+ M$, then $p$ can be viewed as the future limit of a future-directed chain of elements in $I^-(p,\tilde{M}) \cap M$ which is necessarily future-endless as a chain in $M$ as $p \notin M$. To show the converse, just note that if $P:= I^-(p,\tilde{M}) \cap M$ is a TIP in $M$, then $p \notin M$, for otherwise
\[
I^-(p,\tilde{M}) \cap M \equiv I^-(p,M)
\]
from the causal convexity of $M$, and $P$ would then be a PIP in $M$, a contradiction. Thus $p \in \partial ^+ M$, proving Claim 4.

\smallskip

\noindent {\em Claim 5.} $\Psi$ is continuous.\\
Let $(p_n)$ be a sequence in $M \cup\partial ^+ M$ converging to some $p$ therein. We wish to show that ${\rm cl}_M(I^-(p,\tilde{M}) \cap M)$ is the Hausdorff closed limit in $M$ of the sequence $(P_n:=I^-(p_n,\tilde{M}) \cap M))$. The key remark here is that the proof of Theorem \ref{thething} (ii) applied to $(\tilde{M},\tilde{g})$ means that ${\rm cl}_{\tilde{M}}(I^-(p,\tilde{M}))$ is actually the Hausdorff closed limit of $(I^-(p_n,\tilde{M}))$ in $\tilde{M}$. (Note that $p_n \rightarrow p$ in $\tilde{M}$ since we are using the induced topology.)

Let $x \in \limsup (P_n)$. Given any $\tilde{U} \ni x$ open set in $\tilde{M}$, $\tilde{U}\cap M$ is a open set in $M$, and hence infinitely many $P_n$'s intersect it. In other words, infinitely many $I^-(p_n,\tilde{M})$'s intersect $\tilde{U}$. Therefore, $x \in \limsup (I^-(p_n,\tilde{M}))$ in $\tilde{M}$, and from the key remark above $x \in {\rm cl}_{\tilde{M}}(I^-(p,\tilde{M}))\cap M$. However, since $M$ is open in $\tilde{M}$, any open set in $M$ containing $x$ is also open in $\tilde{M}$, and hence must intersect $I^-(p,\tilde{M})\cap M$. We conclude that $x \in {\rm cl} _M(I^-(p,\tilde{M})\cap M)$.

Now, let $y \in {\rm cl}_M(I^-(p,\tilde{M})\cap M) \subset {\rm cl}_{\tilde{M}}(I^-(p,\tilde{M}))$, and $U \ni y$ open set in $M$. Since $U$ is also open in $\tilde{M}$, the key remark above now implies that $x \in \liminf (I^-(p_n,\tilde{M}))$ in $\tilde{M}$, and hence there exists $n_0 \in \mathbb{N}$ such that $U$ intersects $I^-(p_n,\tilde{M})$ for all $n \in \mathbb{N}$ larger than $n_0$. Then $U \cap P_n \neq \emptyset$ for all large enough $n$, i.e., $y \in \liminf (P_n)$ as desired.  This completes the proof of Claim 5.

\smallskip

\noindent {\em Claim 6.} $\Psi ^{-1}$ is continuous.\\
Let $(p_n)$ be a sequence in $M \cup\partial ^+ M$ and $p \in  M \cup\partial ^+ M$, such that ${\rm cl}_M(I^-(p,\tilde{M}) \cap M)$ is the Hausdorff closed limit in $M$ of the sequence $(I^-(p_n,\tilde{M}) \cap M)$. We now wish to show that ${\rm cl}_{\tilde{M}}(I^-(p,\tilde{M}))$ is the Hausdorff closed limit of $(I^-(p_n,\tilde{M}))$ in $\tilde{M}$, because then Theorem \ref{thething} (ii) applied to $(\tilde{M},\tilde{g})$ will imply that $p_n \rightarrow p$ in $\tilde{M}$. This in turn will establish Claim 6.

Let $q \in \limsup (I^-(p_n,\tilde{M}))$ in $\tilde{M}$. Then for some subsequence $(p_k)$ of $(p_n)$ we can pick $(q_k)$ in $\tilde{M}$ converging to $q$ such that
\[
q_k \in I^-(p_k,\tilde{M}), \forall k.
\]
 A future-directed timelike curve from $q_k$ to $p_k$ must intersect $M$ by Lemma \ref{technical}, and hence we can choose $r_k \in I^-(p_k,\tilde{M}) \cap M$ with $q_k \ll_{\tilde{g}} r_k$, for each $k$. Pick any $q'\ll_{\tilde{g}} q$. Let $K\subset \tilde{M}$ be a compact set whose causal past in $(\tilde{M},\tilde{g})$ contains $M$, and hence also $M\cup \partial ^+ M$. Eventually, $r_k$ enters the compact set $J^+(q', \tilde{M}) \cap J^-(K,\tilde{M})$. We may then consider, up to passing to a subsequence, that $(r_k)$ converges to some $r \in \liminf(I^-(p_n,\tilde{M}) \cap M) \equiv {\rm cl}_M(I^-(p,\tilde{M}) \cap M) \subset J^-(p,\tilde{M})$. Closedness of the causal relation in $(\tilde{M},\tilde{g})$ then implies that $q \in J^-(p,\tilde{M}) \equiv {\rm cl}_{\tilde{M}}(I^-(p,\tilde{M}))$.

Now, let $x \in {\rm cl}_{\tilde{M}}(I^-(p,\tilde{M})) \equiv J^-(p,\tilde{M})$ be given. Let $\tilde{U} \ni x$ be an open set in $\tilde{M}$, and pick any $x' \in \tilde{U} \cap I^-(x, \tilde{M})$. Then $x' \in I^-(p,\tilde{M})$, and again any future-directed timelike curve from $x'$ to $p$ must intersect $M$ by Lemma \ref{technical}. We may then assume that there exists $y \in I^-(p,\tilde{M})\cap M$ such that $x' \ll _{\tilde{g}} y$. Since $I^+(x') \cap M$ is open in $M$ and contains $y$, for large enough $n \in \mathbb{N}$ we have $I^-(p_n,\tilde{M}) \cap M \cap I^+(x') \neq \emptyset$, and hence $x'\in I^-(p_n,\tilde{M})$, that is, $I^-(p_n,\tilde{M})\cap \tilde{U} \neq \emptyset$. We conclude that $x \in \liminf (I^-(p_n,\tilde{M}))$ in $\tilde{M}$. Therefore, ${\rm cl}_{\tilde{M}}(I^-(p,\tilde{M}))$ is indeed the Hausdorff closed limit of $(I^-(p_n,\tilde{M}))$ in $\tilde{M}$.

\smallskip

$(3)$ Given any future-inextendible null geodesic $\gamma:[0,A) \rightarrow M$ ($A\leq +\infty$) in $(M,g)$, it is a null pregeodesic in $(\tilde{M},\tilde{g})$, and hence admits a reparametrization $\tilde{\gamma}:[0,a) \rightarrow \tilde{M}$ as a null geodesic therein. Let $K\subset \tilde{M}$ be a compact set whose causal past in $(\tilde{M},\tilde{g})$ contains $M$. Now, $\tilde{\gamma}$ cannot be future-inextendible in $\tilde{M}$, since it is contained in the compact set $C=J^+(\gamma(0),\tilde{M}) \cap J^-(K, \tilde{M})$. (In particular, $a< +\infty$.) Therefore, $\tilde{\gamma}$ has a future-inextendible geodesic extension which must leave $C$ and hence $M$. It must therefore intersect $\partial^+ M$; the intersection point is of course a future endpoint of $\gamma$, and the proof is complete.
\end{proof}

\begin{remark}
\label{def:3}
{\em This theorem admits an alternative proof by comparing the CLT with the chronological topology (cf. Remark \ref{rem:pruebasalt} in the Appendix).}
\end{remark}




\vspace{.2cm}

\section{Null infinity for c-boundaries: a revision}\label{Scriplus}

We wish to revisit here the recently introduced notion of {\em null infinity}
in the context of this paper. These notions were introduced in full generality in Ref. \cite{Costa_e_Silva_2018}, and in particular the full $c$-completion obtained by adding both the past and the future $c$-boundaries was used. As a consequence, some technicalities were needed in order to ensure a good behavior of the appropriate objects.


Our aim here is to show that in the more restricted situation we contemplate here (to wit, globally hyperbolic spacetimes with only future $c$-boundaries), some of the aforementioned technicalities naturally disappear, allowing a more natural treatment. The first non-trivial aspect is that the global hyperbolicity ensures that the future and the past $c$-boundaries are completely independent \cite[Theorem 9.1]{Florescausalboundaryspacetimes2007}, \cite[Theorem 3.29]{Floresfinaldefinitioncausal2011}, and no essential information is missed by considering partial boundaries. In addition, we will see that the enhanced causality allows further simplifications even after discarding the past $c$-boundary.

\subsection{On the definition of null infinity for future c-boundaries}

Let us begin by adapting the notion of null infinity \cite[Definition 4.4]{Costa_e_Silva_2018} to the (partial) future $c$-boundary. We restrict ourselves to a globally hyperbolic spacetime $(M,g)$ without boundary for simplicity, although all the notions in this section are also applicable for the case with timelike boundary with few modifications.

\begin{definition}
\label{def:1}
The {\em future null infinity} of $(M,g)$, denoted by $\mathcal{J}^+$, is the set of elements $P\in \hat{M}$ such that:
\begin{enumerate}[label=(\roman*)]
\item \label{def:1-1} $\exists$ a future-complete null ray $\eta:[0,\infty)\rightarrow M$ with $P=I^-(\eta)$.
\item \label{def:1-2} Every future-inextendible null geodesic $\gamma$ for which $P=I^-(\gamma)$ is future-complete.
\end{enumerate}
\end{definition}

A couple of remarks about this definition are in order:
\begin{itemize}
\item[1)] Observe that the main difference between this definition and the one in \cite{Costa_e_Silva_2018} is that we do not require here the notion of \textit{future regularity}. Indeed, future regularity is required to avoid certain pathological behaviors the endpoints of curves may have, which only occur when the full $c$-completion is considered. These issues are altogether absent for partial boundaries.
\item[2)] It is an easy exercise to check that the condition $P=I^-(\alpha)$ for a future-inextendible causal curve $\alpha:[a,b) \rightarrow M$ is equivalent to the following statement: for any sequence $(t_k)$ in $[a,b)$ converging to $b$, we have $ i_M(\alpha(t_k)) (\equiv I^+(\alpha(t_k))) \rightarrow P$ in the CLT. The latter statement is precisely the meaning of $P$ being a future-endpoint of $\alpha$ on $\hat{\partial}M$.
\end{itemize}

\smallskip

With the notion of null infinity for $c$-boundaries in place, we proceed to consider its relation to {\em conformal} null infinity. The definition of conformal null infinity we adopt here is essentially the one given in \cite[Definition 4.1]{Costa_e_Silva_2018}, which we reproduce here for convenience:

\begin{definition}
Let $(\tilde{M},\tilde{g})$ be a conformal extension of the spacetime $(M,g)$ with future conformal boundary $\partial^{+} M$ and conformal factor $\Omega$. A point $p \in \partial^{+} M$ is said to be {\em at infinity} if there exist an open set $\tilde{U} \ni p$ of $\tilde{M}$ and a $C^{\infty}$ extension $\tilde{\Omega}$ of $\Omega$ to $M \cup \tilde{U}$ such that $\tilde{\Omega}(p) =0$. Such a point at infinity is said to be {\em regular} if in addition $d\tilde{\Omega}(p) \neq 0$. The set of all regular points at infinity is the {\em conformal (future) null infinity} ({\em associated with the conformal extension} $(\tilde{M},\tilde{g})$) and will be denoted by $\mathcal{J}^{+}_c$.
\end{definition}

As expected, when both a conformal null infinity and a null infinity on the $c$-boundary exist, they are closely related. However, the null infinity on $c$-boundaries as defined is generally larger than the conformal one\footnote{This is to be expected, because the conformal null infinity requires a conformal factor $\Omega$ which extends in a $C^{\infty}$ fashion onto the boundary points, which cannot be expected to happen in general. But if it does, then any point in $\mathcal{J}^{+}$ also belongs to $\mathcal{J}_{c}^{+}$ due to \cite[Theorem 4.3 (ii)]{Costa_e_Silva_2018}.}.

\begin{proposition}\label{voila}
  Let $(\tilde{M},\tilde{g})$ be a future-nesting conformal extension of a globally hyperbolic spacetime $(M,g)$, and consider the homeomorphism $\Psi:M \cup \partial ^+ M\rightarrow \hat{M}$, given in the proof of Theorem \ref{causaltoconformal}  taking $\partial^{+}M$ onto $\hat{\partial} M$. Then,

  \[\Psi(\mathcal{J}_{c}^{+})\subset \mathcal{J}^{+}.\]

\end{proposition}
\begin{proof}
  Let $p \in \mathcal{J}^+_c$. From \cite[Theorem 4.3 (iii)]{Costa_e_Silva_2018}, we know that $p$ is the future endpoint in $M\cup \partial^+ M$ of a future-complete null ray $\eta:[0,+\infty) \rightarrow M$. Let us denote by $P=I^-(\eta)$ and recall that, from the definition of the Hausdorff closed limit, $P\in \hat{\partial} M$ is the future-endpoint of the curve $\eta$ on $\hat{M}$. Thus, $P$ satisfies clause \textit{\ref{def:1-1}} of Definition \ref{def:1}. To show that it also satisfies the clause \textit{\ref{def:1-2}}, consider any future-inextendible null geodesic $\alpha:[0,A) \rightarrow M$ ($A\leq +\infty$) such that $I^-(\alpha) = P$. As the map $\Psi$ is an homeomorphism, then $p=\Psi^{-1}(P)$ is the endpoint of $\alpha$ on $\hat{M}_i$,  and hence $\alpha$ must be future-complete by \cite[Theorem 4.3 (i)]{Costa_e_Silva_2018}. Thus, $P$ also satisfies clause \textit{\ref{def:1-2}} as claimed, and the proof is complete.
\end{proof}

Therefore, under mild hypothesis on the conformal extension, the conformal null infinity is (usually properly) included in the $c$-boundary null infinity.

\subsection{Regularity of the null infinity}

Once the notion of null infinity is translated to the setting of this paper, all the results in \cite[Section 5]{Costa_e_Silva_2018} follow with few modifications. However, it is important to note that the existence of a conformal boundary often adds good properties to the future null infinity $\mathcal{J}^{+}$. As an example, consider the notion of \emph{regularity} for the $c$-boundary null infinity given in \cite[Definition 5.6]{Costa_e_Silva_2018}, and adapted here for future $c$-boundaries:

\begin{definition}
  \label{ample}
  The future null infinity $\mathcal{J}^+$ is said to be:
  \begin{itemize}
  \item[(A1)] {\em Ample}\footnote{This definition, as per \cite[Definition 5.6]{Costa_e_Silva_2018}, uses the {\em chronological} topology, {\em not} the CLT. However, in the specific results we prove here, we will see that these topologies coincide, and hence the distinction is not important.} if for any compact set $C \subset M$, and
    for any connected component $\mathcal{J}^+_0$ of $\mathcal{J}^+$,
    $\mathcal{J}^+_0\cap (\overline{M} \setminus \widetilde{I^+ (C)})$ is
    a non-empty open set, where
    \begin{equation}
      \widetilde{I^+(C)}:=\{P\in \hat{M}: I^-(x)\subset P \mbox{ for some $x\in C$}\}.\label{eq:2}
    \end{equation}
    (In intuitive terms: no connected component of $\mathcal{J}^+$
    can be ``entirely contained'' in the future of a compact
    set.)

  \item[(A2)] {\em Past-complete} if for any $P\in \mathcal{J}^+$
    and $P' \in \hat{\partial} M$ with $P' \subset P$, then
    $P'\in \mathcal{J}^+$. (In other words: any element of the
    future $c$-boundary which is in the causal past of $\mathcal{J}^+$ must also
    be in $\mathcal{J}^+$.)

\end{itemize}
We say that the future null infinity $\mathcal{J}^+$ is
\emph{regular} if it is both ample and past-complete.
\end{definition}


The regularity of null infinity in the sense just defined was crucial in \cite{Costa_e_Silva_2018} in order to obtain analogues of key properties for black holes in the conformal context. For instance, \cite[Corollary 5.9]{Costa_e_Silva_2018} provided a generalization for $c$-boundaries of the classic result that closed trapped surfaces are inside the black hole region when $(M,g)$ satisfies the {\em null convergence condition} ($Ric(v,v) \geq 0$ for any null vector $v \in TM$). Moreover, as shown in \cite[Appendix]{Costa_e_Silva_2018}, this regularity condition was unavoidable, as counterexamples to the result appear if one removes it.

The next result goes a step further towards justifying this regularity condition. Indeed, it shows that regularity in this sense is guaranteed in many physically interesting cases where the conformal boundary is present.

\begin{proposition}
  \label{voila2}
  Let $(\tilde{M},\tilde{g})$ be a future-nesting conformal
  extension of the globally hyperbolic spacetime $(M,g)$. Assume in
  addition that $\mathcal{J}^+_c$ is a smooth null hypersurface in
  $(\tilde{M},\tilde{g})$ whose null geodesic generators are
  past-inextendible in $I^+(M,\tilde{M})$. Then:
  \begin{enumerate}[label=(\alph*)]
  \item \label{auxitem-1} For any compact set $C \subset M$, and for any connected
    component $\mathcal{I}^+_0$ of $\mathcal{J}_c^+$,
    $\mathcal{I}^+_0\cap [(\partial^+M \cup M) \setminus I^+ (C,
    \tilde{M})]\neq \emptyset$.

  \item \label{auxitem-2} If $p \in \partial^+M$, $q \in \mathcal{J}_c^+$, and
    $p \leq _{\tilde{g}} q$, then $p \in \mathcal{J}_c^+$.
  \item \label{auxitem-3} If the future null infinity $\mathcal{J}^+$ is
    connected\footnote{Recall that in view of Theorem \ref{causaltoconformal}, the CLT and the topology induced by that of $\tilde{M}$ coincide.}, then it is regular.
  \end{enumerate}
\end{proposition}

\begin{proof}
  \textit{\ref{auxitem-1}} Given any $p \in \mathcal{I}^+_0$, the past-directed null
  generator of $\mathcal{J}^+_c$ starting at $p$ must eventually leave
  the compact set
  $J^+(C,\tilde{M})\cap J^-(p,\tilde{M}) \subset I^+(M,\tilde{M})$,
  and will then leave $I^+ (C, \tilde{M})$ while still contained
  inside $\partial ^+M$.

  \smallskip

  \textit{\ref{auxitem-2}} If $p \in \partial^+M$, $q \in \mathcal{J}_c^+$, and
  $p \leq _{\tilde{g}} q$, let $\alpha:[0,1] \rightarrow \tilde{M}$ be
  a past-directed causal curve from $q$ to $p$. Since $\partial^+M$ is
  achronal (cf. Corollary \ref{convexitygh3}), $\alpha$ must {\em
    initially} (i.e., near $q$) be, up to reparametrization, a piece
  of the past-directed null generator of $\mathcal{J}^+_c$ starting at
  $q$. Since this is past-inextendible in $I^+(M,\tilde{M})$, it
  cannot leave the connected component of $\mathcal{J}^+_c$ containing
  $q$, and hence $\alpha$ must be {\em entirely} contained in such a
  generator. Therefore $p \in \mathcal{J}^+_c$ as well.

  \smallskip

\textit{\ref{auxitem-3}} The first point we note is that Remark \ref{rem:pruebasalt} in the Appendix implies that the CLT and the chronological topology coincide here, so there is no ambiguity in the use of topological terms and we may as well think in terms of either the induced topology or the CLT here. In particular, $\hat{M}$ is Hausdorff and $\widetilde{I^{+}(C)}$ is closed by \cite[Proposition 5.7]{Costa_e_Silva_2018}.

In order to see that $\mathcal{J}^{+}$ is indeed regular, we need to show that it is both ample and past-complete. The latter is quite straightforward once we recall that both $\hat{M}$ and $M\cup \partial^+M$ are chronologically isomorphic (cf. Theorem \ref{causaltoconformal}) and also assertion \textit{\ref{auxitem-2}}). Hence, we focus on the ample condition.

Let $C$ be a compact set. Observe that we need to prove that
\begin{equation}
\mathcal{J}^{+}\cap \left( \hat{M}\setminus \widetilde{I^{+}(C)} \right)\neq \emptyset,\label{eq:5}
\end{equation}
as the openness follows from the fact that $\widetilde{I^{+}(C)}$ is closed. Assume by way of contradiction that \eqref{eq:5} is empty, so that $\mathcal{J}^{+}\subset \widetilde{I^+(C)}$. Since $\Psi(\mathcal{J}_{c}^{+})\subset \mathcal{J}^{+}$ by Proposition \ref{voila}, from item \textit{\ref{auxitem-1}} we conclude that
\begin{equation}
  \mathcal{J}^{+}\cap \left( \hat{M}\setminus I^{+}(K) \right)\neq \emptyset\label{eq:6}
\end{equation}
for any compact set $K$. Let $X$ be a past-directed timelike vector field defined on a neighbourhood $U$ of $C$, and denote by $\varphi_{X}:W\subset U\times \mathbb{R}\rightarrow M$ its associated flow. Since $C$ is compact, there exists $t_{0}>0$ such that $\varphi^{t_{0}}_{X}(x):=\varphi_{X}(x,t_{0})$ is defined for all $x\in C$.

Define then $C^{t_{0}}=\varphi^{t_{0}}_{X}(C)$, which is a compact set satisfying, by the timelike character of $X$, that $\widetilde{I^{+}(C)}\subset I^{+}(C^{t_{0}})$. But then $\mathcal{J}^{+}\subset \widetilde{I^+(C)}\subset I^+(C^{t_0})$ and hence
\[
  \mathcal{J}^{+}\cap \left( \hat{M}\setminus I^{+}(C^{t_{0}}) \right)= \emptyset,
  \]
a contradiction with \eqref{eq:6}.

\end{proof}

\section{Appendix: CLT vs chronological to\-po\-lo\-gy}
\label{sec:comp-clt-chron}

As announced in the Introduction, in this Appendix we compare two topologies defined on the future c-completion $\hat{M}$, namely the CLT and the chronological topology (see \cite{HarrisTopologyfuturechronological2000, Florescausalboundaryspacetimes2007, Floresfinaldefinitioncausal2011}). Now, the reader should observe that although most good properties of the CLT depend on global hyperbolicity its definition still makes sense for strongly causal spacetimes; we will need only assume strongly causality for the results below.
Therefore, throughout this section $(M,g)$ shall denote a strongly causal spacetime without boundary.

As a first step, the chronological topology is defined formally in terms of the so-called chronological limit. There are, however, some technicalities we have to deal with before giving the definition. We begin by slightly changing the approach given above for the CLT topology, adapting it to previous studies on the chronological topology.

As with the Hausdorff closed limit, the chronological limit makes use of the concepts of \textit{inferior} and \textit{superior} limits of sets. However, there is a key difference between both approaches: namely, the former focuses on closed subsets of $M$ while the latter works in terms of open ones. In order to unify both approaches, let us define

\begin{equation}
  \label{eq:1}
  \begin{array}{c}
    \Limsup(A_n) := \{x \in M \, : \, \mbox{ $x$ belongs to infinitely many $A_n$'s}\},\\
    \Liminf(A_n) := \{x \in M \, : \, \mbox{  $x$ belongs to all but finitely many $A_n$'s}\}.
  \end{array}
\end{equation}
(compare with \eqref{eq:infsup}). Then, for a sequence $\sigma=(P_n)$ of (open) indecomposable past sets, we can define the {\em Hausdorff  (open) limit}

\[
\Lhaus(\sigma)=\{P\in \hat{M}:P=\Liminf (\sigma))=\Limsup (\sigma)\}.
  \]
 This limit defines a sequential topology on $\hat{M}$. It is a simple exercise to prove that this sequential topology coincides with the CLT topology on $\hat{M}$ upon examining the proof of Theorem \ref{CLT} \ref{CLT-2}.

  \smallskip

  Now, let us introduce the notion of \textit{chronological limit} $\Lchr$: for any sequence $\sigma=(P_{n})$ of IPs, consider $\Lchr$ given by:

\begin{equation}
P\in \Lchr(\sigma) \iff \left\{
  \begin{array}{l}
    P\subset \Liminf (\sigma)\\
    P \hbox{ is a maximal IP on $\Limsup (\sigma)$}.
  \end{array}\right.\label{eq:comparison1}
\end{equation}
As in the case of the Hausdorff open limit, the chronological limit defines a sequential topology $\tau_{chr}$ called the (future) chronological topology. This topology is not necessarily Hausdorff even when $(M,g)$ is globally hyperbolic (see the grapefruit-on-a-stick example in \cite[Example 4]{HarrisCausalboundarystandard2001}), but it does possess the properties listed in Theorem \ref{thething} even for strongly causal spacetimes (this follows for instance from the proof of \cite[Thm. 3.27]{Floresfinaldefinitioncausal2011}.)

\smallskip

Let us point out some basic relations between the CLT topology (defined in terms of the Hausdorff open topology) and the chronological topology. First, let us observe that the requirements for $P$ to be in the chronological limit of a sequence $\sigma$ are weaker than in the Hausdorff open limit case. Therefore, we can easily state that $\Lhaus(\sigma)\subset \Lchr(\sigma)$ for any sequence $\sigma\subset \hat{M}$. In particular, the manner sequential topologies are defined (see Proposition \ref{finetop}) implies that the chronological topology is coarser than the CLT topology.

It also follows that both limits coincide when the sequence $\sigma=(P_{n})$ is a chronological chain on $M$. In such a case $P_{n}=I^-(x_n)$ for some $x_n\in M$ and $P_n\subset P_{n+1}$ for all $n$. In particular, $\Liminf(P_n)=\Limsup(P_n)$.

Although the two topologies differ even in globally hyperbolic spacetimes, one might expect that there are natural situations where these topologies are the same. The next result closes the gap between them characterizing precisely when they coincide:

    \begin{proposition}\label{prop:comp1}
    Let $\hat{M}$ be a future completion endowed with the limit operators $\Lhaus$ and $\Lchr$. Both limit operators coincide if, and only if, the chronological topology is Hausdorff.
    \end{proposition}

    \begin{proof}
      For the implication to the right, just recall that $\tau_c=\tauC$ if $\Lhaus=\Lchr$, and $\Lhaus$ is Hausdorff.


      For the implication to the left, let us assume that $\tauC$ is Hausdorff, and let us prove that both limits coincide. By contradiction, let us assume that there exists $P\in \Lchr(\sigma)\setminus \Lhaus(\sigma)$ for some sequence $\sigma\subset \hat{M}$. In particular, it should follows from the definition of $\Lhaus$ that $\Liminf(\sigma)\subsetneq \Limsup(\sigma)$. Let $x\in \Limsup(\sigma)\setminus \Liminf(\sigma)\neq\emptyset$ and consider $P'\in \hat{M}$ the maximal IP in the superior limit with $x\in P'$ (which does exist due to Zorn's Lemma).

      Let us analyze the properties of $P'$. On the one hand, from construction, it cannot coincide with previous $P$, as it contains the point $x$ and
      \[x\in \Limsup(\sigma)\setminus \Liminf(\sigma)\subset \Limsup(\sigma)\setminus P\]
      (for the last inequality recall that $P\in \Lchr(\sigma)$, and thus, $P\subset \Liminf(\sigma)$.) On the other hand, there exists a subsequence $\sigma'\subset \sigma$ such that $P'\in \Lchr(\sigma')$. In fact, let us assume that $\sigma=(P_n)$, and denote by $\{p_n\}_n$ a chronological future chain defining $P'$. Since $P'\subset \Limsup(\sigma)$, we can take a subsequence $(n_k)$ such that $p_k\in P_{n_k}$ for all $k\in \mathbb{N}$, i.e., $P'\subset \Liminf(\sigma')$, with $\sigma'=(P_{n_k})$. Due the maximality of $P'$ on the superior limit of $\sigma$, it follows that $P'\in \Lchr(\sigma')$.

      Finally, let us show that such a subsequence leads us to a contradiction. Observe that as $P\in \Lchr(\sigma)$, then both $P$ and $P'$ are two different TIPs belonging to $\Lchr(\sigma')$, which contradicts the Hausdorff hypothesis for $\tauC$, as required.
    \end{proof}

\smallskip

\begin{remark}
  \label{rem:pruebasalt}
\em  The previous comparison yields alternative proofs for Theorems \ref{thething} and \ref{causaltoconformal} whenever it applies. In fact, Theorem \ref{thething} can be obtained in this setting by using \cite[Theorem 3.27]{Floresfinaldefinitioncausal2011}, taking into account the fact that global hyperbolicity ensures that no TIP and TIF are related in the (total) c-completion \cite[Theorem 3.29]{Floresfinaldefinitioncausal2011} (thus both the future and the past c-completions can be treated separately.)

  For Theorem \ref{causaltoconformal}, just observe that in the hypotheses of that theorem $M\subset \tilde{M}$ is causally convex, so it follows from \cite[Theorem 4.16]{Floresfinaldefinitioncausal2011} that the conformal and the $c$-completion coincide chronologically, as well as topologically, when the latter is endowed with the chronological topology. In particular, the chronological topology is Hausdorff in this case, and thus it coincides with the CLT by Proposition \ref{prop:comp1}.
\end{remark}

\section*{Acknowledgments}

The authors are partially supported by the Spanish Grant MTM2016-78807-C2-2-P (MINECO and FEDER funds).


\end{document}